\documentclass[nolayout]{article}

\usepackage[utf8]{inputenc}
\usepackage[english]{babel}
\usepackage{amsmath,amsthm}
\usepackage{geometry}
\usepackage[textwidth=4cm,textsize=footnotesize]{todonotes}
\usepackage{fancyhdr}
\usepackage{amssymb,color,bbm,xargs}
\usepackage{graphicx}
\usepackage[active]{srcltx}
\usepackage{ifthen}
\usepackage{enumerate}
\usepackage{color}
\usepackage{accents}
\usepackage{dsfont}
\usepackage[ruled,vlined]{algorithm2e}
\usepackage{subfig}
\usepackage{multirow}
\usepackage{placeins}

\usepackage{aliascnt}
\theoremstyle{plain}

\newtheorem*{assumption*}{A}

\newaliascnt{proposition}{theorem}

\aliascntresetthe{proposition}

\newaliascnt{lemma}{theorem}
\newtheorem{lemma}[lemma]{Lemma}
\aliascntresetthe{lemma}
\newaliascnt{corollary}{theorem}

\aliascntresetthe{corollary}

\theoremstyle{definition}
\newaliascnt{definition}{theorem}

\aliascntresetthe{definition}

\newaliascnt{remark}{theorem}

\aliascntresetthe{remark}

\usepackage{hyperref}
\usepackage[nameinlink,noabbrev]{cleveref}

\def\nset{\mathbb{N}}

\newcommandx{\PEt}[2][1=]{\mathbb{E}_{#1}\left[#2\right]}
\def\rmd{\mathrm{d}}
\def\rme{\mathrm{e}}

\def\dimy{\mathsf{p}}
\def\dimz{\mathsf{m}}

\def\barG{\overline{G}}
\def\barH{\overline{H}}

\def\eqsp{\,}

\newcommand{\normMat}[2]{\left\|#2\right\|_{#1}}

\newcounter{hypH}

\newcommand{\eqdef}{\ensuremath{:=}}

\newcommandx\sequence[3][2=,3=]
{\ifthenelse{\equal{#3}{}}{\ensuremath{\{ #1_{#2}\}}}{\ensuremath{( #1_{#2}, \eqsp #2 \in #3 )}}}
\newcommandx{\sequencen}[2][2=n\in\nset]{\ensuremath{( #1, \eqsp #2 )}}
\newcommandx\dsequence[4][3=k,4=\zset]{\ensuremath{( (#1_{#3}, #2_{#3}), \eqsp #3 \in #4 )}}
\newcommandx{\as}[1][1=\PP]{\ensuremath{#1\--\mathrm{a.s.}}}

\begin{document}

\title{Particle rejuvenation of Rao-Blackwellized Sequential Monte Carlo smoothers for Conditionally Linear and Gaussian models}
\date{}

\author{Ngoc Minh Nguyen\footnote{LTCI, CNRS and T\'el\'ecom ParisTech, 46 rue Barrault 75634 Paris Cedex 13, France.}\and Sylvain {L}e {C}orff\footnote{Laboratoire de Math\'ematiques d'Orsay, Univ. Paris-Sud, CNRS, Universit\'e Paris-Saclay, 91405 Orsay, France.}\and Eric Moulines\footnote{Centre de Math\'ematiques Appliqu\'ees, UMR 7641, Ecole Polytechnique, France.}}

\lhead{Ngoc Minh Nguyen et al.}
\rhead{Conditionally Linear and Gaussian models}

\maketitle

\begin{abstract}
This paper focuses on Sequential Monte Carlo approximations of smoothing distributions in conditionally linear and Gaussian state spaces. To reduce Monte Carlo variance of smoothers, it is typical in these models to use Rao-Blackwellization: particle approximation is used to sample sequences of hidden regimes while the Gaussian states are explicitly integrated conditional on the sequence of regimes and observations, using variants of the Kalman filter / smoother. The first successful attempt to use Rao-Blackwellization for smoothing extends the Bryson-Frazier smoother for Gaussian linear state space models using the generalized two-filter formula together with Kalman filters / smoothers. More recently, a  forward backward decomposition of smoothing distributions mimicking the Rauch-Tung-Striebel smoother for the regimes combined with backward Kalman updates has been introduced.  
This paper investigates the benefit of introducing additional rejuvenation steps in all these algorithms to sample at each time instant new  regimes conditional on the forward and backward particles. This defines particle based approximations of the smoothing distributions whose support is not restricted to the set of particles sampled in the forward or backward filter. These procedures are applied to commodity markets which are described using a two factor model based on the spot price and a convenience yield for crude oil data.
\end{abstract}

\section{Introduction} \label{sec:intro}
State space models are bivariate stochastic processes $\{(Y_i,Z_i)\}_{i\ge 1}$ where the state sequence $(Z_i)_{i\ge 1}$ is a Markov chain which is only partially observed through the sequence $(Y_i)_{i\ge 1}$. Conditionally on the state sequence $(Z_i)_{i\ge 1}$ the observations are independent and for all $\ell\ge 1$ the conditional distribution of $Y_{\ell}$ given $(Z_i)_{i\ge 1}$ depends on $Z_{\ell}$ only. These models are used in a large variety of disciplines such as financial econometrics, biology, signal processing, see \cite{delmoral:2013} and the references therein. In general state space models, bayesian filtering and smoothing problems, i.e. the computation of the posterior distributions of a sequence of states $(Z_{i},\ldots,Z_{p})$ for $1\le i\le p\le \ell$ given observations $(Y_{1},\ldots,Y_{\ell})$, are challenging tasks. Filtering refers to the estimation of the distributions of the hidden state $Z_i$ given the observations $(Y_1,\ldots,Y_i)$ up to time $i$, while fixed-interval smoothing stands for the estimation of the distribution of sequence of states $(Z_{i},\ldots,Z_{p})$ given observations $(Y_{1},\ldots,Y_{\ell})$ with $1\le i\le p<\ell$. 
 
When the state and observation models are linear and Gaussian, filtering can be solved explicitly using the Kalman filter \cite{kalman:1960}. Exact solutions of the fixed-horizon smoothing problem can be obtained using either the Rauch-Tung-Striebel smoother \cite{rauch:striebel:tung:1965} or the Bryson-Frazier two-filter smoother \cite{bryson:frazier:1963}. 
This paper focuses on  Conditionally Linear and Gaussian Models (CLGM) given for $i\ge 2$ by:
\begin{equation}
\label{eq:model:state}
Z_i = d_{a_{i}} + T_{a_{i}}Z_{i-1}+H_{a_{i}}\varepsilon_i\eqsp,
\end{equation}
where:
\begin{enumerate}[-]
\item $(\varepsilon_i)_{i\ge 2}$ is a sequence of independent and identically distributed (i.i.d.) $\dimz$-dimensional Gaussian vectors with zero mean and identity covariance.
\item $(a_i)_{i\ge 1}$ is a homogeneous Markov chain taking values in a finite space $\{1,\ldots,J\}$, called \emph{regimes}, with initial distribution $\pi$ and transition matrix $Q$.  
\item $(H_j)_{1\le j\le J}$ are $\dimz \times \dimz$ positive-definite matrices, $(d_j)_{1\le j\le J}$ $\dimz$-dimensional vectors and $(T_j)_{1\le j\le J}$ $\dimz \times \dimz$ positive-definite matrices.
\item $Z_1$ is a $\dimz$-dimensional Gaussian random variable with mean $\mu_1$ and variance $\Sigma_1$ independent of $(\varepsilon_i)_{i\ge 2}$.
\end{enumerate}
Let $n$ be the number of observations. At each time step $1\le i \le n$, the observation $Y_i$ is given by:
\begin{equation}
\label{eq:model:obs}
Y_i = c_{a_i} + B_{a_i}Z_i + G_{a_i}\eta_i\eqsp,
\end{equation}
where:
\begin{enumerate}[-]
\item $(\eta_i)_{i\ge 1}$ is a i.i.d. sequence of $\dimy$-dimensional Gaussian vectors, independent of $(\varepsilon_i)_{i\ge 2}$ and $Z_1$.
\item $(G_j)_{1\le j \le J}$ are $\dimy \times \dimy$ positive-definite matrices, $(c_j)_{1\le j \le J}$ $\dimy$-dimensional vectors and $(B_j)_{1\le j \le J}$ $\dimy \times \dimz$ matrices.
\end{enumerate}
CLGM play an important role in many applications; see \cite{sarkka:2013} and the references therein for an up-to-date account. A crucial feature of these models is that, conditional on the regime sequence $(a_1,\ldots,a_n)$, both the state equation and the observation equation are linear and Gaussian, which implies that conditional on the sequence of regimes and on the observations, the filtering and the smoothing distributions of the continuous states $(Z_1,\dots,Z_n)$ can be computed explicitly.

To exploit this specific structure, it has been suggested in the pioneering works of \cite{chen:liu:2000,doucet:godsill:andrieu:2000} to solve the filtering problem by combining Sequential Monte Carlo (SMC) methods to sample the regimes with the Kalman filter to compute the conditional distribution of the states sequence $(Z_i)_{1\le i\le n}$ conditional on the regimes and on the observations.
This is a specific instance of Rao-Blackwellized Monte Carlo filters, often referred to as the \emph{Mixture Kalman Filter}. Improvements of these early filtering techniques  have been introduced in \cite{doucet:gordon:krishnamurthy:2001,schon:gustafsson:nordlund:2005}. 

The use of Rao-Blackwellization to solve the smoothing problem has been proved to be more challenging and has received satisfactory solutions only recently. The first forward-backward smoother proposed in the literature \cite{fong:godsill:doucet:west:2002} was not fully Rao-Blackwellized as it required to sample the hidden linear states in the backward pass. An alternative approach, based on the so-called structural approximation of the model suggested in an early paper by \cite{kim:1994}, was proposed by \cite{barber:2006} to avoid to sample a continuous state in the backward pass. This approximation is rather ad-hoc and the resulting smoother is not consistent when the number of particles goes to infinity. The inaccuracy introduced by the approximation might be difficult to control.

The first fully Rao-Blackwellized SMC smoother which should lead to consistent approximations when the number of particles grows to infinity was proposed by  \cite{briers:doucet:maskell:2010} and extends the Bryson-Frazier smoother for Gaussian linear state space models using the generalized two-filter formula with Rao-Blackwellization steps for the forward and the backward filters.  This two-filter approach combines a forward filter with a backward information filter which are approximated numerically using SMC for the regime sequence and Kalman filtering techniques for the hidden linear states.

More recently, \cite{sarkka:bunch:godsill:2012,lindsten:bunch:godsill:schon:2013,lindsten:bunch:sarkka:schon:godsill:2015} introduced a Rao-Blackwellized smoother based on the forward-backward decomposition of the FFBS algorithm with Rao-Blackwellization steps both in the forward and backward time directions. The update of the smoothing distribution of the regime given the observations shares some striking similarities with the Rauch-Tung-Striebel smoothing procedure, which is at the heart of the FFBS procedure. The Rao-Blackwellization requires to update backward in time the smoothing distribution of the states given the regimes and the observations, which is achieved by using an \emph{\`a la Kalman} backward update.

In this paper, we propose to improve the performance of the algorithms introduced in \cite{briers:doucet:maskell:2010} and in \cite{sarkka:bunch:godsill:2012,lindsten:bunch:godsill:schon:2013,lindsten:bunch:sarkka:schon:godsill:2015} by using additional Rao-Blackwellization steps which allows to sample new particles in the backward pass. This approach may be seen as an extension of the ideas of \cite{fearnhead:clifford:2003} for Rao-Blackwellized smoothers.  In \cite{briers:doucet:maskell:2010}, for all $1\le i\le n$, the sampled forward and backward sequences are merged to approximate the posterior distribution of $(a_i,z_i)$. This provides an approximation whose support is restricted to the particles produced at time $i$ by the backward particle filter. As noted in \cite[Secion~2.6]{fearnhead:wyncoll:tawn:2010}, these two-filter smoothers are prone to suffer from degeneracy issues when the algorithm associates forward particles at time $i-1$ with backward particles at time $i$. We propose to approximate the marginal smoothing distribution of $(a_i,z_i)$ by merging the sampled forward and backward trajectories at times $i-1$ and $i+1$ and integrating out all possible paths between time $i-1$ and time $i$ and between time $i$ and time $i+1$ instead of sampling random variables. Similarly, in the backward pass of \cite{sarkka:bunch:godsill:2012,lindsten:bunch:godsill:schon:2013,lindsten:bunch:sarkka:schon:godsill:2015}, a regime $\tilde{a}_i$ is sampled at time $1\le i\le n-1$ using the particles produced by the forward filter at time $i$. In this case, particle rejuvenation may be introduced by using the forward weighted samples at time $i-1$ and extending these trajectories at time $i$ with a Kalman filter for all possible values of the regime. Then, $\tilde{a}_i$ is sampled in $\{1,\ldots,J\}$ using an appropriately adapted weight.

The paper is organized as follows. The algorithms introduced in \cite{briers:doucet:maskell:2010} and in \cite{sarkka:bunch:godsill:2012,lindsten:bunch:godsill:schon:2013,lindsten:bunch:sarkka:schon:godsill:2015} as long as the proposed rejuvenation associated with each method are presented in Section~\ref{sec:RaoBlackwell}. The performance of all these methods is illustrated in Section~\ref{sec:numerical:experiments} with simulated data. In Section~\ref{sec:crudeoil}, an application to commodity markets is presented; the performance of our procedure is illustrated with crude oil  data. A detailed derivation of the algorithms is provided in the Appendix.

\section{Rao-Blackwellized smoothing algorithms}
\label{sec:RaoBlackwell}
This section details the Sequential Monte Carlo algorithms which can be used to approximate the conditional distribution of the states $(a_1,\ldots,a_n)$ or the marginal distributions of $(a_i,z_{i})$ given the observations $(Y_1,\ldots,Y_n)$. For all $\dimz\,\times\,\dimz$ matrix  let $|A|$ be the determinant of $A$. If $A$ is a positive-definite matrix, for all $z\in\mathbb{R}^{m}$ define
\[
\normMat{A}{z}^2 \eqdef z'A^{-1}z\eqsp,
\]
where for any vector or matrix $z$, $z'$ denotes the transpose matrix of $z$. Let $m_{}(a_{i},z_{i-1};z_i)$ be the probability density of the conditional distribution of $Z_i$ given $(a_{i},Z_{i-1})$ and $g_{}(a_{i},z_{i};y_i)$ be the probability density of the conditional distribution of $Y_i$ given $(a_{i},Z_{i})$:
\begin{align}
\label{eq:definition-m}
m_{}(a_{i},z_{i-1};z_i) & \eqdef \left|2\pi\barH_{a_i}\right|^{-1/2}\exp\left\{-\frac{1}{2}\normMat{\barH_{a_i}}{z_i -d_{a_{i}} -T_{a_{i}}z_{i-1}}^2\right\}\eqsp,\\
\label{eq:definition-g}
g_{}(a_{i},z_{i};y_i) & \eqdef \left|2\pi\barG_{a_i}\right|^{-1/2}\exp\left\{-\frac{1}{2}\normMat{\barG_{a_i}}{y_i - c_{a_i} - B_{a_i}z_{i}}^2\right\}\eqsp,
\end{align}
where
\[
\barG_j \eqdef G_jG'_j\;,\; \barH_{j} \eqdef H_jH'_j\eqsp.
\]
All the algorithms considered in this paper are based on forward-backward or two-filter decompositions of the smoothing distributions and share the same forward filter presented in Section~\ref{sec:forward}.

\subsection{Forward filter}
\label{sec:forward}
The SMC approximation $p^N_{}(a_{1:i},z_{i}|y_{1:i})$ of $p_{}(a_{1:i},z_{i}|y_{1:i})$ may be obtained using a standard Rao-Blackwellized  algorithm. The procedure produces a sequence of trajectories $(a^k_{1:i})_{1\le k \le N}$ associated with normalized importance weights $(\omega^k_{i})_{1\le k \le N}$ ($\sum_{k=1}^N \omega^k_i = 1$) used to define the following approximation of $p_{}(a_{1:i},z_{i}|y_{1:i})$:
\begin{equation}
\label{eq:forward:SMC}
p^N_{}(a_{1:i},z_{i}|y_{1:i}) = \sum_{k=1}^N\omega^k_{i}\,p_{}(z_{i}|a^k_{1:i},y_{1:i})\,\delta_{a^k_{1:i}}(a_{1:i})\eqsp,
\end{equation}
where $\delta$ is the Dirac delta function. In this equation, the conditional distribution of the hidden state $z_{i}$ given the observations $y_{1:i}$ and a trajectory $a^k_{1:i}$ is a Gaussian distribution whose mean $\mu^k_{i}$ and variance $P^k_{i}$ may be obtained by using the Kalman filter update.

\vspace{.2cm}

\noindent\textbf{Initialization}\\
At time $i=1$, write, for all $1\le j \le J$,
\[
\mu^j_{1|0} = c_{j}+B_{j}\mu_1\;\;\mbox{and}\;\;P_{1|0}^j =B_{j}\Sigma_1B'_{j} + \barG_{j}\eqsp.
\]
$(a^k_1)_{1\le k \le N}$ are sampled independently in $\{1,\ldots,J\}$ with probabilities proportional to
\[
\pi_j p(a_1=j|y_1) \propto \pi_j |P_{1|0}^j|^{-1/2}\exp\left\{-(y_1-\mu^j_{1|0})'(P_{1|0}^j)^{-1}(y_1-\mu^j_{1|0})/2\right\}\eqsp.
\]
Then, $\mu_1^k$ and $P_1^k$ are computed using a Kalman filter:
\begin{align*}
K^k_{1} &=\Sigma_1B'_{a_1^k}\left(B_{a_1^k}\Sigma_1B'_{a_1^k} + \barG_{a_1^k}\right)^{-1}\eqsp,\\
\mu^k_{1} &= \mu_{1} + K^k_1\left(Y_1 - c_{a_{1}^k} - B_{a_1^k}\mu_{1}\right)\eqsp,\\
P^k_{1} &=\left(I_{\dimz}-K^k_{1}B_{a_1^k}\right)\Sigma_1\eqsp,
\end{align*}
where for all positive integer $p$, $I_p$ is the $p\times p$ identity matrix. Each particle particle $a_1^k$ is associated with the importance weight $\omega^k_1 = 1/N$.

\vspace{.2cm}

\noindent\textbf{Iterations}\\
Several procedures may be used to extend the trajectories $(a^k_{1:i-1})_{1\le k \le N}$ at time $i$.
For all sampled trajectories $(a_{1:i-1}^k)_{1\le k \le N}$ and all $1\le j \le J$,  \cite{chen:liu:2000} used the incremental weights:
\[
\gamma_i^{j,k} = p(y_i | a_i = j, a_{1:i-1}^k, y_{1:i-1}) Q(a_{i-1}^k,j)\eqsp.
\]
The conditional distribution of $Y_i$ given $a^k_{1:i-1}$, $a_i$ and $Y_{1:i-1}$ is a Gaussian distribution with mean $c_{a_i}+B_{a_i}\mu^k_{i|i-1}(a_i)$ and variance $B_{a_i}P^k_{i|i-1}(a_i)B'_{a_i} + \barG_{a_i}$ where 
\begin{align*}
\mu^{k}_{i|i-1}(a_i) &= d_{a_i} + T_{a_i}\mu^k_{i-1}\eqsp,\\
P^{k}_{i|i-1}(a_i) &= T_{a_i}P^k_{i-1}T'_{a_i} + \barH_{a_i}\eqsp.
\end{align*} 
Therefore,
\[
\gamma_i^{j,k} \propto  Q(a_{i-1}^k,j)|B_jP^{j,k}_{i|i-1}B'_j + \barG_j|^{-1/2}\exp\left\{-\frac{1}{2}\normMat{B_jP^{j,k}_{i|i-1}B'_j + \barG_j}{y_i-c_{j}-B_j\mu^{j,k}_{i|i-1}}^2\right\} \eqsp,
\]
where:  
\begin{align}
\mu^{j,k}_{i|i-1} &= \mu^{k}_{i|i-1}(j) = d_{j} + T_{j}\mu^k_{i-1}\eqsp,\label{eq:mui|i-1}\\
P^{j,k}_{i|i-1} &= P^{k}_{i|i-1}(j) = T_{j}P^k_{i-1}T'_{j} + \barH_{j}\eqsp.\label{eq:Pi|i-1}
\end{align} 
In \cite{chen:liu:2000}, for all $1\le k \le N$, an ancestral path is chosen with probabilities proportional to  $(\omega^k_{i-1})_{1\le k \le N}$. Then, the new regime $a_i^k$ is sampled in  $\{1,\ldots, J\}$ with probabilities proportional to $(\gamma_i^{j,k})_{1\le j\le J}$. A drawback of this method is that only ancestral paths that have been selected using the importance weights $(\omega^k_{i-1})_{1\le k \le N}$ are extended at time $i$. Following \cite{barembruch:garivier:moulines:2008}, this may be improved by considering all the offsprings of all ancestral trajectories $(a_{1:i-1}^k)_{1\le k \le N}$. Each ancestral path has $J$ offsprings at time $i$, it is thus necessary to choose a given number of trajectories at time $i$ (for instance $N$) among the $NJ$ possible paths. To obtain the weight associated with each offspring write the following approximation of  $p(a_{1:i}|y_{1:i})$ based on the weighted samples at time $i-1$:
\begin{align*}
p^N(a_{1:i}|y_{1:i})&\propto \sum_{k=1}^N\omega^k_{i-1}Q(a^k_{i-1},a_i)p(y_i|a^k_{1:i-1},a_i,y_{1:i-1})\delta_{a^k_{1:i-1}}(a_{1:i-1})\eqsp,\\
&\propto \sum_{k=1}^N\sum_{j=1}^J\omega^k_{i-1}\gamma_i^{j,k}\delta_{(a^k_{1:i-1},j)}(a_{1:i})\eqsp.
\end{align*}
Therefore, each ancestral trajectory of the form  $(a^k_{1:i-1},j)$, $1\le k \le N$, $1\le j\le J$, is associated with the normalized weight $\tilde{\omega}^{j,k}_{i} \propto \omega^k_{i-1}\gamma_i^{j,k}$. Several random selection schemes have been proposed to discard some of the possible offsprings to maintain an average number of $N$ particles at each time step. Following \cite{barembruch:garivier:moulines:2008}, we might choose between the Kullback-Leibler Optimal Selection (KL-OS) or the Chi-Squared Optimal Selection (CS-OS) to associate a new weight to each of the $NJ$ trajectories. If the new weight is 0, then the corresponding particle can be removed.

\vspace{.2cm}

\noindent\textbf{KL-OS:} $\lambda$ is chosen as the solution of :
\[
\sum_{k=1}^N\sum_{j=1}^J\mathrm{min}\left(\tilde{\omega}^{j,k}_{i}/\lambda,1\right) = N\eqsp.
\]
For all $1\le j \le J$ and $1\le k \le N$, if $\tilde{\omega}^{j,k}_{i}\ge \lambda$ then the new weight $\tilde{\Omega}^{j,k}_{i}$ is $\tilde{\Omega}^{j,k}_{i}=\tilde{\omega}^{j,k}_{i}$ and if $\tilde{\omega}^{j,k}_{i}< \lambda$:
\[
\tilde{\Omega}^{j,k}_{i}=
\left\{
 \begin{array}{rl}
  &\hspace{-.5cm}\lambda \;\mbox{with probability } \tilde{\omega}^{j,k}_{i}/\lambda\eqsp,\\
&\hspace{-.5cm}0 \;\mbox{with probability } 1-\tilde{\omega}^{j,k}_{i}/\lambda\eqsp.
\end{array}
\right.
\]
\textbf{CS-OS:} $\lambda$ is chosen as the solution of :
\[
\sum_{k=1}^N\sum_{j=1}^J\mathrm{min}\left(\sqrt{\tilde{\omega}^{j,k}_{i}/\lambda},1\right) = N\eqsp.
\]
For all $1\le j \le J$ and $1\le k \le N$, if $\tilde{\omega}^{j,k}_{i}\ge \lambda$ then the new weight $\tilde{\Omega}^{j,k}_{i}$ is $\tilde{\Omega}^{j,k}_{i}=\tilde{\omega}^{j,k}_{i}$ and if $\tilde{\omega}^{j,k}_{i}< \lambda$:
\[
\tilde{\Omega}^{j,k}_{i}=
\left\{
 \begin{array}{rl}
  &\hspace{-.5cm}\sqrt{\tilde{\omega}^{j,k}_{i}\lambda} \;\mbox{with probability } \sqrt{\tilde{\omega}^{j,k}_{i}/\lambda}\eqsp,\\
&\hspace{-.5cm}0 \;\mbox{with probability } 1-\sqrt{\tilde{\omega}^{j,k}_{i}/\lambda}\eqsp.
\end{array}
\right.
\]
Then, in both cases, all particles such that $\tilde{\Omega}^{j,k}_{i} = 0$ are discarded and for all the other trajectories defined as an ancestral path $(a^k_{1:i-1})$ extended by $a^k_i = j$, the new corresponding weight $\omega$ in \eqref{eq:forward:SMC} is given by the normalized weight $\tilde{\Omega}^{j,k}_{i}$. In the numerical sections of this paper, the Kullback-Leibler Optimal Selection (KL-OS) scheme is used.


\subsection{FFBS based algorithms}
\subsubsection{FFBS algorithms of \cite{sarkka:bunch:godsill:2012,lindsten:bunch:godsill:schon:2013,lindsten:bunch:sarkka:schon:godsill:2015}}
%
%
\cite{sarkka:bunch:godsill:2012,lindsten:bunch:godsill:schon:2013,lindsten:bunch:sarkka:schon:godsill:2015} proposed a Rao-Blackwellized procedure to sample the regime backward in time following the same steps as in the Forward Filtering Backward Smoothing algorithm \cite{huerzeler:kunsch:1998,doucet:godsill:andrieu:2000}. The algorithm relies on the decomposition given, for all $1\le i \le n-1$, by:
\[
p(a_{1:n}|y_{1:n}) = p(a_{1:i}|a_{i+1:n},y_{1:n})p(a_{i+1:n}|y_{1:n})\eqsp.
\]
This decomposition is similar to the Rauch-Tung-Striebel decomposition of the filtering distribution. The first factor on the right hand side of the previous equation is nevertheless more difficult to handle because it itself relies on all the observations. As noted by \cite{sarkka:bunch:godsill:2012}, this term can be computed recursively by considering the following decomposition:
\begin{equation}
\label{eq:FFBS:decomp}
p(a_{1:i}|a_{i+1:n},y_{1:n}) \propto p(y_{i+1:n},a_{i+1:n}|a_{1:i},y_{1:i})p(a_{1:i}|y_{1:i})\eqsp.
\end{equation}
The second factor in the last equation may be approximated using the ancestral  trajectories $(a^k_{1:i})_{1\le k \le N}$ and the associated importance weights $(\omega^k_{i})_{1\le k \le N}$ produced by the forward filter. Therefore, $p(a_{1:i}|a_{i+1:n},y_{1:n})$ may be approximated by:
\[
p^N(a_{1:i}|a_{i+1:n},y_{1:n}) = \sum_{k=1}^N \tilde{\omega}^k_{i|n}\delta_{a^k_{1:i}}(a_{1:i}) \quad \text{with} \quad \tilde{\omega}^k_{i|n} \propto \omega_i^k p(y_{i+1:n},a_{i+1:n}|a^k_{1:i},y_{1:i}) \eqsp.
\]
Then, a trajectory $\tilde{a}_{1:n}$ approximatively distributed according to $p(a_{1:n}|y_{1:n})$ may be sampled following these steps:
\begin{enumerate}[-]
\item Set $\tilde{a}_n= a_n^k$ with probabilities proportional to $(\omega_n^k)_{1\le k \le N}$.
\item For all $1\le i\le n-1$, set $\tilde{a}_i = a_i^k$ with probabilities proportional to $(\tilde{\omega}_{i|n}^k)_{1\le k \le N}$.
\end{enumerate}
This algorithm requires to compute the quantity  $p(y_{i+1:n},a_{i+1:n}|a^k_{1:i},y_{1:i})$. This predictive quantity is available analytically using  Kalman filtering techniques. However, this has to be done for each trajectory $(a^k_{1:i})_{1\le k \le N}$, which leads to an algorithm with a prohibitive computational complexity. \cite{lindsten:bunch:sarkka:schon:godsill:2015} proposed a procedure computationally less intensive by conditioning with respect to $z_i$ and then marginalizing with respect to this variable:
\begin{equation}
\label{eq:FFBS:incremental}
p(y_{i+1:n},a_{i+1:n}|a^k_{1:i},y_{1:i}) = \int p(y_{i+1:n},a_{i+1:n}|z_i,a^k_{i})p(z_{i}|a^k_{1:i},y_{1:i})\rmd z_i\eqsp.
\end{equation}
This is similar to the two-filter decomposition of the smoothing distribution, see Section~\ref{sec:two-filer}. By  \cite{lindsten:bunch:sarkka:schon:godsill:2015},
\[
p(y_{i+1:n},a_{i+1:n}|z_i,a_{i}) \propto Q(a_i,a_{i+1}) \exp\left\{-(z_i'\Omega_i(a_{i+1:n})z_i-2\lambda_i'(a_{i+1:n})z_i)/2\right\}\eqsp,
\]
where the proportionality is with respect to $(a_i,z_i)$ and
\[
p(y_{i:n},a_{i+1:n}|z_i,a_{i}) \propto \exp\left\{-(z_i'\widehat{\Omega}_i(a_{i:n})z_i-2\widehat{\lambda}'_i(a_{i:n})z_i)/2\right\}\eqsp,
\]
where the proportionality is with respect to $z_i$. These quantities may be computed recursively backward in time with:
\begin{align*}
\widehat{\Omega}_n(a_n) &= B'_{a_n}\barG^{-1}_{a_n}B_{a_n}\eqsp,\\
\widehat{\lambda}_n(a_n) &=B'_{a_n}\barG^{-1}_{a_n}(y_n-c_{a_n})\eqsp.
\end{align*}
Then, for $1\le i\le n-1$, define $m_{i+1} = \widehat{\lambda}_{i+1} - \widehat{\Omega}_{i+1}d_{a_{i+1}}$ and $M_{i+1} = H_{a_{i+1}}'\widehat{\Omega}_{i+1}H_{a_{i+1}} + I_{\dimz}$ and write
\begin{align*}
\Omega_i(a_{i+1:n}) &= T'_{a_{i+1}}(I_{\dimz}-\widehat{\Omega}_{i+1}(a_{i+1:n})H_{a_{i+1}}M^{-1}_{i+1}H'_{a_{i+1}})\widehat{\Omega}_{i+1}(a_{i+1:n})T_{a_{i+1}}\eqsp,\\
\lambda_i(a_{i+1:n}) &=T'_{a_{i+1}}(I_{\dimz}-\widehat{\Omega}_{i+1}(a_{i+1:n})H_{a_{i+1}}M^{-1}_{i+1}H'_{a_{i+1}})m_{i+1}\eqsp.
\end{align*}
As $p(y_{i:n},a_{i+1:n}|z_i,a_{i}) = p(y_{i}|z_i,a_{i})p(y_{i+1:n},a_{i+1:n}|z_i,a_{i})$,
\begin{align*}
\widehat{\Omega}_i(a_{i:n}) &= \Omega_i(a_{i+1:n})+ B'_{a_i}\barG^{-1}_{a_i}B_{a_i}\eqsp,\\
\widehat{\lambda}_i(a_{i:n}) & = \lambda_i(a_{i+1:n}) + B'_{a_i}\barG^{-1}_{a_i}(y_i-c_{a_i})\eqsp.
\end{align*}
Then, by \eqref{eq:FFBS:incremental},
\begin{equation}
\label{eq:post:obs}
p(y_{i+1:n},a_{i+1:n}|a^k_{1:i},y_{1:i})\propto Q(a_i^k,a_{i+1})|\Lambda^k_i(a_{i+1:n})|^{-1/2}\exp\left\{-\eta^k_i(a_{i+1:n})/2\right\}\eqsp,
\end{equation}
where the proportionality is with respect to $a^k_{1:i}$ and
\begin{align*}
\Lambda^k_i(a_{i+1:n})&= (\Gamma_i^k)'\Omega_i(a_{i+1:n})\Gamma_i^k + I_{\dimz}\eqsp,\\
\eta^k_i(a_{i+1:n}) &= \|\mu_i^k\|^2_{\Omega^{-1}_i(a_{i+1:n})} - 2\lambda'_i(a_{i+1:n})\mu_i^k-\|(\Gamma_i^k)'(\lambda_i(a_{i+1:n})-\Omega_i(a_{i+1:n})\mu_i^k)\|^2_{\Lambda_i(a_{i+1:n})}\eqsp,
\end{align*}
where $P_i^k = \Gamma_i^k(\Gamma_i^k)'$. Therefore,
\[
\tilde{\omega}_{i|n} \propto \omega_i^kQ(a_i^k,a_{i+1})|\Lambda^k_i(a_{i+1:n})|^{-1/2}\exp\left\{-\eta^k_i(a_{i+1:n})/2\right\}\eqsp.
\]
If $(\tilde{a}^k_{1:n})_{1\le k \le \tilde{N}}$ are independent copies of $\tilde{a}_{1:n}$, the SMC approximation of \cite{lindsten:bunch:sarkka:schon:godsill:2015} of the joint smoothing distribution of the regime is:
\[
p^{\mathsf{Lbscg}}_{\tilde{N}}(a_{1:n}|Y_{1:n}) = \frac{1}{\tilde{N}}\sum_{k=1}^{\tilde N} \delta_{\tilde{a}^{k}_{1:n}}(a_{1:n})\eqsp.
\]

\subsubsection{Particle rejuvenation of FFBS algorithms}
The crucial step of the FFBS algorithm is the decomposition \eqref{eq:FFBS:decomp} which allows to extend a backward trajectory $\tilde{a}_{i+1:n}$ by choosing a particle in the set $(a_i^k)_{1\le k \le N}$ produced by the forward filter (and discarding the states $a^k_{1:i-1}$). An improved version of this FFBS algorithm which is not constrained to sample states in the support $(a_i^k)_{1\le k \le N}$ may be defined for all $2\le i\le n-1$ by writing:
\begin{align*}
p(a_{1:i}|a_{i+1:n},y_{1:n}) &\propto p(y_{i+1:n},a_{i+1:n}|a_{1:i},y_{1:i})p(a_{1:i}|y_{1:i})\eqsp,\\
&\propto p(y_{i+1:n},a_{i+1:n}|a_{1:i},y_{1:i})\int p(a_{1:i-1},z_{i-1}|y_{1:i-1})Q(a_{i-1},a_i)\\
&\hspace{6cm}m(a_i,z_{i-1};z_i)g(a_i,z_i;y_i)\rmd z_{i-1:i}\eqsp.
\end{align*}
Replacing $p(a_{1:i-1},z_{i-1}|y_{1:i-1})$ in the integral by the particle approximation obtained during the forward pass and using Kalman filtering techniques for each trajectory $(a^k_{1:i-1})_{1\le k\le N}$ and each $a_i\in\{1,\ldots,J\}$ yields:
\[
\int p^N(a_{1:i-1},z_{i-1}|y_{1:i-1})Q(a_{i-1},a_i)m(a_i,z_{i-1};z_i)g(a_i,z_i;y_i)\rmd z_{i-1:i} \propto \sum_{k=1}^N \omega_{i|i-1}^k(a_i)\delta_{a^k_{1:i-1}}(a_{1:i-1})\eqsp,
\]
where
\[
\omega_{i|i-1}^k(a_i) = \omega_{i-1}^k Q(a_{i-1}^k,a_i)|\Sigma^k_{i|i-1}(a_i)|^{-1/2}\mathrm{exp}\left\{-\frac{1}{2}\|y_i - y^k_{i|i-1}(a_i)\|_{\Sigma^k_{i|i-1}(a_i)}\right\}\eqsp,
\]
\[
y^k_{i|i-1}(a_i) = c_{a_i} + B_{a_i}(d_{a_i}+T_{a_i}\mu^k_{i-1})\;\mbox{and}\; \Sigma^k_{i|i-1}(a_i) = B_{a_i}(T_{a_i}P^k_{i-1}T'_{a_i}+\barH_{a_i})B'_{a_i} + \barG_{a_i}\eqsp.
\]
On the other hand, for all $1\le k \le N$, $p(y_{i+1:n},a_{i+1:n}|a^k_{1:i-1},a_i,y_{1:i})$ is computed as in \eqref{eq:post:obs} with all possible values $a_i\in\{1,\ldots,J\}$ and not only the regime of the filtering pass $(a_i^k)_{1\le k \le N}$. This means that a Kalman filter must be used for each trajectory $a^k_{1:i-1}$ which may be extended by $a_i\in\{1,\ldots,J\}$. Denote by $\mu_{i|i-1}^k(a_i)$ and $P_{i|i-1}^k(a_i)$ the mean and covariance matrix of the law of $z_i$ given $(a^k_{1:i-1},a_i)$ obtained as in \eqref{eq:mui|i-1} and \eqref{eq:Pi|i-1}. Then,
\begin{equation}
\label{eq:posterior:FFBS}
p(y_{i+1:n},a_{i+1:n}|a^k_{1:i-1},a_i,y_{1:i}) = Q(a_i,a_{i+1})|\Lambda^k_{i|i-1}(a_{i:n})|^{-1/2}\exp\left\{-\eta^k_{i|i-1}(a_{i:n})/2\right\}\eqsp,
\end{equation}
where the proportionality is with respect to $(a^k_{1:i-1},a_i)$ and
\begin{align*}
\Lambda^k_{i|i-1}(a_{i:n})&= (\Gamma_{i|i-1}^k(a_i))'\Omega_i(a_{i+1:n})\Gamma_{i|i-1}^k(a_i) + I_{\dimz}\eqsp,\\
\eta^k_{i|i-1}(a_{i:n}) &= \|\mu_{i|i-1}^k(a_i)\|^2_{\Omega^{-1}_i(a_{i+1:n})} - 2\lambda'_i(a_{i+1:n})\mu_{i|i-1}^k(a_i)\\
&\hspace{3cm}-\|(\Gamma_{i|i-1}^k(a_i))'(\lambda_i(a_{i+1:n})-\Omega_i(a_{i+1:n})\mu_{i|i-1}^k(a_i))\|^2_{\Lambda_i(a_{i+1:n})}\eqsp,
\end{align*}
where $\Gamma_{i|i-1}^k(a_i)$ is defined as $P_{i|i-1}^k(a_i) = \Gamma_{i|i-1}^k(a_i)(\Gamma_{i|i-1}^k(a_i))'$. The distribution $p(a_{1:i}|a_{i+1:n},y_{1:n})$ is then approximated by :
\begin{multline}
\label{eq:FFBS:rejuv}
p^N(a_{1:i}|a_{i+1:n},y_{1:n})\\
\propto\sum_{k=1}^N \omega_{i|i-1}^k(a_i)Q(a_i,a_{i+1})|\Lambda^k_{i|i-1}(a_{i:n})|^{-1/2}\exp\left\{-\eta^k_{i|i-1}(a_{i:n})/2\right\}\delta_{a^k_{1:i-1}}(a_{1:i-1})\eqsp.
\end{multline}
By integrating over all possible paths $a_{1:i-1}$, $\tilde{a}_i$ is sampled in $\{1,\ldots,J\}$. This particle rejuvenation step allows to explore states which are not in the support of the particles produced by the forward filter  and improves the accuracy and the variance of the original FFBS algorithm, see Section~\ref{sec:numerical:experiments} for numerical illustrations. 

Another modification of the FFBS algorithm based on a Markov chain Monte Carlo (MCMC) sampling step was introduced in \cite[Section~5.2]{lindsten:bunch:sarkka:schon:godsill:2015}. Instead of sampling from \eqref{eq:FFBS:rejuv}, \cite[Section~5.2]{lindsten:bunch:sarkka:schon:godsill:2015} proposed to draw a forward path $a_{1:i-1}$ in $(a^k_{1:i-1})_{1\le k \le N}$ and a sate $a_i$ in $\{1,\ldots,J\}$ according to:
\[
\widetilde{q}(a_{1:i}|a_{i+1:n},y_{1:n}) = \sum_{k=1}^N \widetilde{\vartheta}^k_{i-1}\widetilde{q}(a_{i}|a^k_{1:i-1},a_{i+1:n},y_{1:n})\delta_{a^k_{1:i-1}}(a_{1:i-1})\eqsp,
\]
where $(\widetilde{\vartheta}^k_{i-1})_{1\le k \le N}$ are adjustment multipliers and $\widetilde{q}(a_{i}|a^k_{1:i-1},a_{i+1:n},y_{1:n})$ is a proposal kernel chosen by the user. This means that an ancestral path $a^{\star}_{1:i-1}$ is sampled in $(a^k_{1:i-1})_{1\le k \le N}$ with weights $(\widetilde{\vartheta}^k_{i-1})_{1\le k \le N}$ and $a^{\star}_i$ is sampled from $\widetilde{q}(\cdot|a^{\star}_{1:i-1},a_{i+1:n},y_{1:n})$. Then, the proposed sequence $a^{\star}_{1:i}$ is accepted or rejected using the usual Metropolis-Hastings acceptance ratio. The choice of MCMC rejuvenation has interesting practical consequences as the computation of the acceptance ratio only requires to compute the posterior probability \eqref{eq:posterior:FFBS} for the proposed sequence $a^{\star}_{1:i}$ while our technique is based on the computation of  \eqref{eq:posterior:FFBS} for all combinations of sequences $(a^{k}_{1:i})_{1\le k \le N}$ and states $a_i\in\{1,\ldots,J\}$. Sampling from \eqref{eq:FFBS:rejuv} is computationally more intensive, especially when $N$ is large, but our method is based on a direct approximation of $p(a_{1:i}|a_{i+1:n},y_{1:n})$ based on  $(a^k_{1:i-1})_{1\le k \le N}$ and $a_{i+1:n}$ instead of approximate MCMC draws.

\subsection{Rao-Blackwellized Two-filter Smoother}
\label{sec:two-filer}
\subsubsection{Rao-Blackwellized Two-filter Smoother of \cite{briers:doucet:maskell:2010}}
\label{eq:rbtf}
Contrary to the previous methods, two-filter based smoothers are designed to compute approximations of marginal smoothing distributions (usually the posterior distribution of one or two consecutive  regimes given all the observations). \cite{briers:doucet:maskell:2010} introduced the following decomposition of the smoothing distributions for all $2\le i \le n$:
\begin{equation}
\label{eq:dens:bwd_fwd_doucet}
p(a_{i},z_{i}|y_{1:n})  \propto  p_{}(a_{i},z_{i}|y_{1:i-1})p(y_{i:n}|a_i,z_i)\eqsp.
\end{equation}
The first term on the right hand side may be approximated using the forward filter by noting that:
\[
 p_{}(a_{i},z_{i}|y_{1:i-1}) =\sum_{a_{i-1}}\int_{z_{i-1}}p_{}(a_{i-1},z_{i-1}|y_{1:i-1})m_{}(a_{i},z_{i-1};z_{i})Q(a_{i-1},a_{i}) \rmd z_{i-1}\eqsp.
\]
In the forward pass described in Section~\ref{sec:forward}, a set of possible sequences of regimes $a_{1:i-1}^k$ associated with importance weights $\omega_{i-1}^k$, $1\le k\le N$ are sampled to approximate $p_{}(a_{i-1},z_{i-1}|y_{1:i-1})$. This provides a normalized approximation $p^N(a_{i},z_{i}|y_{1:i-1})$ of $p(a_{i},z_{i}|y_{1:i-1})$. Define
\begin{align*}
\Omega^k_{i-1}(a_{i}) &= T_{a_{i}}P^k_{i-1}T'_{a_{i}} + \barH_{a_{i}}\eqsp,\;\mu^k_{i-1}(a_{i}) = d_{a_{i}} + T_{a_{i}}\mu^k_{i-1}\eqsp,\;r_{i-1}^k(a_{i}) = (\Omega^k_{i-1}(a_{i}))^{-1}\mu^k_{i-1}(a_{i})\eqsp,\\
\omega_{\mathsf{f},i}^k(a_{i}) &= \omega^k_{i-1}Q(a^k_{i-1},a_{i}) \left|2\pi \Omega^k_{i-1}(a_{i})\right|^{-1/2}\exp\left\{-\frac{1}{2}\normMat{\Omega^k_{i-1}(a_{i})}{\mu^k_{i-1}(a_{i})}^2\right\}\eqsp.
\end{align*}
Then,
\begin{equation}
\label{eq:rN}
p^N(a_{i},z_{i}|y_{1:i-1}) = \sum_{k=1}^N \omega_{\mathsf{f},i}^k(a_{i})\exp\left\{-\frac{1}{2}\normMat{\Omega^k_{i-1}(a_{i})}{z_{i}}^2+z'_{i}r_{i-1}^k(a_{i})\right\}\eqsp.
\end{equation}

As the function $(a_{i},z_{i})\mapsto p_{}(y_{i:n}|a_{i},z_{i})$ is not a probability density function, approximating the second term of \eqref{eq:dens:bwd_fwd_doucet}  using SMC samples is not straightforward. The backward filter uses artificial densities to introduce a surrogate target density function which may be approximated recursively using SMC methods. Then, the forward and backward weighted samples are combined using \eqref{eq:dens:bwd_fwd_doucet} to approximate $p(a_{i},z_{i}|y_{1:n})$. Following \cite{briers:doucet:maskell:2010}, for any probability densities $(\gamma^{}_{i})_{1\le i \le n}$, define the following joint probability densities:
\[
\tilde{p}_{n}(a_n,z_n,y_n) \eqdef \gamma^{}_{n}(a_n,z_n)g_{}(a_n,z_n;y_n)\eqsp,\quad \tilde{p}_{n}(y_n)\eqdef \sum_{a_n=1}^J\int \gamma^{}_{n}(a_n,z_n)g_{}(a_n,z_n;y_n)\rmd z_n\eqsp,
\]
and, for all $1\le i\le n-1$,
\begin{align*}
\tilde{p}_{i}(a_{i:n},z_{i:n},y_{i:n}) &\eqdef \gamma^{}_{i}(a_i,z_i)p_{}(y_{i:n}|a_{i:n},z_{i:n})p_{}(a_{i+1:n},z_{i+1:n}|a_i,z_i)\eqsp,\\
\tilde{p}_{i}(y_{i:n})&\eqdef \sum_{a_{i:n}=1}^J\int \gamma^{}_{i}(a_i,z_i)p_{}(y_{i:n}|a_{i:n},z_{i:n})p_{}(a_{i+1:n},z_{i+1:n}|a_i,z_i)\rmd z_{i:n}\eqsp.
\end{align*}
Note that this choice differs slightly from \cite{briers:doucet:maskell:2010} where it is advocated to set $\gamma^{}_i$ as the product of two independent densities $\gamma_i^{a}(a_i)$ and $\gamma_i^{z}(z_i)$. As the accuracy of the algorithm relies heavily on a proper tuning of this artificial density, a more general choice of $\gamma^{}_i$ is considered in this paper. 
By Lemma~\ref{eq:technical:twofilter}, these probability densities may be used to approximate the quantities $p_{}(y_{i:n}|a_{i},z_{i})$, $1\le i \le n$, in \eqref{eq:dens:bwd_fwd_doucet}.
\begin{lemma}
\label{eq:technical:twofilter}
For all $1\le i \le n-1$,
\begin{align}
\tilde{p}_{i}(a_i,z_i|y_{i:n}) &= p_{}(y_{i:n}|a_i,z_i)\gamma^{}_{i}(a_i,z_i)/\tilde{p}_{i}(y_{i:n})\eqsp,\label{eq:p:y1:n}\\
\tilde{p}_{i}(a_{i},z_{i}|y_{i:n}) &= \gamma^{}_i(a_i,z_i)\sum_{a_{i+1:n}=1}^J\frac{\tilde{p}_{i}(a_{i:n}|y_{i:n})p_{}(y_{i:n}|a_{i:n},z_{i})}{\int \gamma^{}_{i}(a_i,z')p_{}(y_{i:n}|a_{i:n},z')\rmd z'}\label{eq:p:ai:zi}\eqsp.
\end{align}
\end{lemma}
\begin{proof}
The proof is postponed to Appendix~\ref{sec:app}.
\end{proof}
By definition of $\tilde{p}_i$ for all $1\le i \le n$,
\begin{align*}
\tilde{p}_{i}(a_{i:n},z_i|y_{i:n}) &\propto \gamma^{}_{i}(a_i,z_i) \int p_{}(y_{i:n}|a_{i:n},z_{i:n})p_{}(a_{i+1:n},z_{i+1:n}|a_i,z_i)\rmd z_{i+1:n}\eqsp,\\
&\propto\gamma^{}_{i}(a_i,z_i)\left\{\prod_{u=i}^{n-1}Q(a_u,a_{u+1})\right\}p(y_{i:n}|z_i,a_{i:n})\eqsp.
\end{align*}
This yields:
\[
\tilde{p}_{i}(a_{i:n}|y_{i:n}) \propto \left\{\prod_{u=i}^{n-1}Q(a_u,a_{u+1})\right\} \int \gamma^{}_{i}(a_i,z_i) p(y_{i:n}|z_i,a_{i:n})\rmd z_i\eqsp.
\]
A set of weighted trajectories $(\tilde{a}^\ell_{i:n})_{1\le \ell\le N}$ with importance weights $(\tilde{\omega}^\ell_i)_{1\le \ell\le N}$, $1\le i \le n$, may then be sampled recursively backward in time to produce a SMC approximation of $\tilde{p}_{}(a_{i:n}|y_{i:n})$ as follows.
\begin{enumerate}[-]
\item For $1\le \ell \le N$, sample $\tilde{a}^{j}_n\sim\tilde{q}_n(\cdot)$ and set:
\[
\tilde{\omega}^{\ell}_n \propto \frac{\int\gamma^{}_n(\tilde{a}^{\ell}_n,z')g_{}(\tilde{a}^{\ell}_n,z';y_n)\rmd z'}{\tilde{q}_n(\tilde{a}^{\ell}_n)}\eqsp.
\]
\item For all $1\le i \le n-1$, resample the set $(\tilde{a}^{\ell}_{i+1:n})_{1\le j\le N}$ using the normalized weights $(\tilde{\omega}^{\ell}_{i+1})_{1\le j \le N}$. Then, for $1\le \ell \le N$, sample $\tilde{a}^{j}_{i}\sim\tilde{q}_i(\tilde{a}^{\ell}_{i+1:n},\cdot)$  and set:
\[
\tilde{\omega}^{\ell}_i \propto \frac{Q(\tilde{a}^{\ell}_i,\tilde{a}^{\ell}_{i+1})\int \gamma^{}_i(\tilde{a}^{\ell}_i,z')p_{}(y_{i:n}|\tilde{a}^{\ell}_{i:n},z')\rmd z'}{\tilde{q}_i(\tilde{a}^{\ell}_{i+1:n},\tilde{a}^{\ell}_i)\int \gamma^{}_{i+1}(\tilde{a}^{\ell}_{i+1},z')p_{}(y_{i+1:n}|\tilde{a}^{\ell}_{i+1:n},z')\rmd z'}\eqsp.
\]
\end{enumerate}
To obtain uniformly weighted samples at each time step, in the numerical experiments we use:
\[
\tilde{q}_n(\cdot) = \int\gamma^{}_n(\cdot,z')g_{}(\cdot,z';y_n)\rmd z'\quad\mbox{and}\quad \tilde{q}_i(\tilde{a}^{\ell}_{i+1:n},\cdot) = \frac{Q(\cdot,\tilde{a}^{\ell}_{i+1})\int \gamma^{}_i(\cdot,z')p_{}(y_{i:n}|(\cdot,\tilde{a}^{\ell}_{i+1:n}),z')\rmd z'}{\int \gamma^{}_{i+1}(\tilde{a}^{\ell}_{i+1},z')p_{}(y_{i+1:n}|\tilde{a}^{\ell}_{i+1:n},z')\rmd z'}\eqsp.
\]
By \eqref{eq:p:y1:n} and \eqref{eq:p:ai:zi}, 
\begin{align*}
p_{}(y_{i:n}|a_i,z_i)
&= \frac{\tilde{p}_{i}(y_{i:n})  \tilde{p}_{i}(a_i,z_i|y_{i:n})}{\gamma^{}_{i}(a_i,z_i)}= \tilde{p}_{i}(y_{i:n})  \sum_{a_{i+1:n}=1}^J\frac{\tilde{p}_{i}(a_{i:n}|y_{i:n})p_{}(y_{i:n}|a_{i:n},z_{i})}{\int \gamma^{}_{i}(a_i,z')p_{}(y_{i:n}|a_{i:n},z')\rmd z'}\label{eq:p:ai:zi}\eqsp,
\end{align*}
which suggests the following particle approximation $p^{N}_{}(y_{i:n}|a_{i},z_{i})$ of $p_{}(y_{i:n}|a_{i},z_{i})$:
\begin{equation}
\label{eq:likelihoodTF}
p^{N}_{}(y_{i:n}|a_{i},z_{i}) = \tilde{p}_{i}(y_{i:n}) \sum_{\ell=1}^N\frac{\tilde{\omega}^{\ell}_i p_{}(y_{i:n}|\tilde{a}^{\ell}_{i:n},z_{i})}{\int \gamma^{}_{i}(\tilde{a}^{\ell}_i,z')p_{}(y_{i:n}|\tilde{a}^{\ell}_{i:n},z')\rmd z'}\delta_{\tilde{a}^{\ell}_{i}}(a_{i})\eqsp.
\end{equation}
The conditional likelihood of the observations given the sequence of states $p_{}(y_{i:n}|a_{i:n},z_{i})$ can be computed explicitly using a Gaussian backward smoother; these computations are summarized in Lemma~\ref{lem:py}.
In the numerical experiments, $\gamma^{}_i(a_i,z_i)$ is set as a mixture of Gaussian distributions. Note that for such a choice, the integral $\int \gamma^{}_{i}(a_i,z')p_{}(y_{i:n}|a_{i:n},z')\rmd z'$ may be computed explicitly, see Lemma~\ref{lem:integral:gammap}. Then, combining \eqref{eq:likelihoodTF} and \eqref{eq:rN} with \eqref{eq:dens:bwd_fwd_doucet} provides an approximation of $p(a_{i},z_{i}|y_{1:n})$ by merging the forward particles $(a^k_{i-1})_{1\le k \le N}$ with the backward particles $(\tilde{a}^k_{i+1})_{1\le k \le N}$, the support of this SMC approximation of $p(a_{i},z_{i}|y_{1:n})$ being $(\tilde{a}^k_{i+1})_{1\le k \le N}$. 

As noted in \cite[Secion~2.6]{fearnhead:wyncoll:tawn:2010}, two-filter smoothers are prone to suffer from degeneracy issues when the algorithm associates forward particles at time $i-1$ with backward particles at time $i$. The authors illustrate this issue in the case where the hidden state is an AR(2) process. To overcome the weakness of such standard two-filter approaches the particle rejuvenation proposed in Section~\ref{sec:tf:rejuv} follows the idea introduced in \cite{fearnhead:wyncoll:tawn:2010} where new particles at time $i$ are sampled conditional on $(a^k_{1:i-1})_{1\le k \le N}$ and on $(\tilde{a}^k_{i+1:n})_{1\le k \le N}$ and appropriately weighted. This allows to produce new particles at time $i$ and to obtain a SMC approximation of $p(a_{i},z_{i}|y_{1:n})$ whose support is not restricted to $(\tilde{a}^k_{i+1})_{1\le k \le N}$. Section~\ref{sec:tf:rejuv} exploits this idea in the specific case of linear and Gaussian models where explicit computations allows to produce an approximation using $(a^k_{1:i-1})_{1\le k \le N}$ and $(\tilde{a}^k_{i+1:n})_{1\le k \le N}$ with support $\{1,\ldots,J\}$ and without any additional sampling steps.

\subsubsection{Particle rejuvenation of two-filter based algorithms}
\label{sec:tf:rejuv}
For $2\le i \le n-1$, particle rejuvenation relies on the explicit marginalization:
\begin{equation}
\label{eq:dens:bwd_fwd}
p(a_{i},z_{i}|y_{1:n})  =  \sum_{a_{i-1}}\sum_{a_{i+1}} \int_{z_{i-1}} \int_{z_{i+1}} \psi^n_{i}(a_{i-1:i+1},z_{i-1:i+1})\rmd z_{i-1}\rmd z_{i+1}\eqsp,
\end{equation}
where $\psi^n_{i}(a_{i-1:i+1},z_{i-1:i+1})$ is the smoothing distribution of the hidden regimes and states between time indices $i-1$ and $i+1$. Note that the EM algorithm requires the approximation of $p(a_{i-1},z_{i-1},a_{i-1},z_{i-1}|y_{1:n})$ in the E-step, this may be obtained following the same steps by marginalizing explicitly the linear states at time $i-2$ and $i+1$. Intermediate computations follow the same steps as  for the approximation of $p(a_{i},z_{i}|y_{1:n})$. First, $\psi^n_{i}$ may be decomposed as follows:
\begin{multline*}
\psi^n_{i}(a_{i-1:i+1},z_{i-1:i+1}) \propto p_{}(y_{i+1:n}|a_{i+1},z_{i+1})p_{}(a_{i-1},z_{i-1}|y_{1:i-1})Q(a_{i-1},a_{i})m_{}(a_{i},z_{i-1};z_{i})\\
\times g_{}(a_{i},z_{i};y_{i})Q(a_{i},a_{i+1})m_{}(a_{i+1},z_{i};z_{i+1})\eqsp,
\end{multline*}
where the proportionality is with respect to $(a_{i-1:i+1},z_{i-1:i+1})$. Then, by \eqref{eq:dens:bwd_fwd}, the smoothing distribution $p(a_{i},z_{i}|y_{1:n}) $ may be written as
\begin{equation}
\label{eq:smoothing}
p(a_{i},z_{i}|y_{1:n})   \propto p(a_{i},z_{i}|y_{1:i-1})g_{}(a_{i},z_{i};y_{i})t_{i}(a_i,z_i,y_{i+1:n})\eqsp,
\end{equation}
where $m$ and $g$ are defined in \eqref{eq:definition-m} and \eqref{eq:definition-g} and
\begin{equation}
t_{i}(a_i,z_i,y_{i+1:n}) = \sum_{a_{i+1}}\int_{z_{i+1}}m_{}(a_{i+1},z_{i};z_{i+1})Q(a_{i},a_{i+1})p_{}(y_{i+1:n}|a_{i+1},z_{i+1}) \rmd z_{i+1}\eqsp.\label{eq:smooth:t}
\end{equation}
The backward pass described in Section~\ref{eq:rbtf} produces a sequence of states $\tilde{a}_{i+1:n}^{\ell}$ associated with importance weights $\tilde{\omega}_{i+1}^{\ell}$, $1\le \ell\le N$ which are used to approximate $p_{}(y_{i+1:n}|a_{i+1},z_{i+1})$. Plugging this approximation into \eqref{eq:smooth:t} provides an approximation $t^N_{i}(a_i,z_i,y_{i+1:n})$ of $t_{i}(a_i,z_i,y_{i+1:n})$ integrating over all possible choices $(a_{i+1},z_{i+1})$. These steps are then combined to form a non normalized SMC approximation of $p(a_{i},z_{i}|y_{1:n})$ using \eqref{eq:smoothing}. The normalization of the SMC approximation of $p(a_{i},z_{i}|y_{1:n})$ is obtained by integrating over the states $a_{i},z_{i}$, when $p(a_{i},z_{i}|y_{1:i-1})$ and $t_{i}(a_i,z_i,y_{i+1:n})$ are replaced by $p^N(a_{i},z_{i}|y_{1:i-1})$ and $t^N_{i}(a_i,z_i,y_{i+1:n})$ in \eqref{eq:smoothing}. 
Our procedure allows to construct sequence of regimes with non-degenerated importance weights in the combination step. This procedure improves significantly \cite{briers:doucet:maskell:2010} where no marginalization of $p(a_{i},z_{i}|y_{1:n})$ over the states at times $i-1$ and $i+1$ is performed and where the proposed forward and backward paths are directly merged. This method often leads to importance weights which are close to be numerically degenerated. By Lemma~\ref{lem:py}, the SMC approximation $p^{N}_{}(y_{i:n}|a_{i},z_{i})$ of $p_{}(y_{i:n}|a_{i},z_{i})$ is then given by:
\begin{equation}
\label{eq:pN:yi:n}
p^{N}_{}(y_{i:n}|a_{i},z_{i}) = \tilde{p}_i(y_{i:n}) \sum_{\ell=1}^N  \frac{\delta_{\tilde{a}^{\ell}_i}(a_i)\tilde{\omega}_{i}^{\ell}}{\int \gamma^{}_{i}(\tilde{a}^{\ell}_i,z')p_{}(y_{i:n}|\tilde{a}^{\ell}_{i:n},z')\rmd z'} \exp\left\{- \frac{1}{2}\normMat{\tilde{P}^{\ell}_{i}}{z_i}^2 + z'_i\tilde{\nu}^{\ell}_{i} - \frac{1}{2} \tilde{c}^{\ell}_i \right\}\eqsp,
\end{equation}
where $(\tilde{P}^{\ell}_{i})^{-1} \eqdef \tilde{P}_{i}^{-1}(\tilde{a}^{\ell}_{i:n})$, $\tilde{\nu}^{\ell}_{i}\eqdef \tilde{\nu}_{i}(\tilde{a}^{\ell}_{i:n})$ and $\tilde{c}^{\ell}_{i} \eqdef \tilde{c}^{\ell}_{i}(\tilde{a}^{\ell}_{i:n})$ are defined in Lemma~\ref{lem:py}. Define
\begin{align*}
\Delta^{\ell}_{i+1} &\eqdef \left(I_{\dimz} + H'_{\tilde{a}^\ell_{i+1}}(\tilde{P}^{\ell}_{i+1})^{-1}H_{\tilde{a}^\ell_{i+1}}\right)^{-1}\eqsp,\\
\delta^{\ell}_{i+1}&\eqdef\tilde{\nu}^{\ell}_{i+1} + \barH_{\tilde{a}^\ell_{i+1}}^{-1}(d_{\tilde{a}^\ell_{i+1}}+T_{\tilde{a}^\ell_{i+1}}z_i)\eqsp.
\end{align*}
Then, by \eqref{eq:smooth:t}, the SMC approximation $t^N_{i}(a_i,z_i,y_{i+1:n})$ of $t_{i}(a_i,z_i,y_{i+1:n})$ is given by:
\begin{align}
t^N_{i}(a_i,z_i,y_{i+1:n}) &=\sum_{a_{i+1}=1}^J\int_{z_{i+1}}m_{}(a_{i+1},z_{i};z_{i+1})Q(a_{i},a_{i+1})p^N_{}(y_{i+1:n}|a_{i+1},z_{i+1}) \rmd z_{i+1}\eqsp,\nonumber\\
&=  \tilde{p}_{i+1}(y_{i+1:n}) \sum_{\ell=1}^N C_i^{-1}(\tilde{a}^{\ell}_{i+1:n}) Q(a_i, \tilde{a}^{\ell}_{i+1}) \tilde{\omega}^{\ell}_{i+1} |\barH_{\tilde{a}^{\ell}_{i+1}}|^{-1/2} |H_{\tilde{a}^{\ell}_{i+1}} \Delta^{\ell}_{i+1} H'_{\tilde{a}^{\ell}_{i+1}}|^{1/2} \eqsp\nonumber \\
&\hspace{2cm} \times \exp\left\{\frac{1}{2} (\delta^{\ell}_{i+1})'H_{\tilde{a}^{\ell}_{i+1}}\Delta^{\ell}_{i+1}H'_{\tilde{a}^{\ell}_{i+1}} \delta^{\ell}_{i+1} -\frac{1}{2}\normMat{\barH_{\tilde{a}^{\ell}_{i+1}}}{d_{\tilde{a}^{\ell}_{i+1}}+T_{\tilde{a}^{\ell}_{i+1}}z_i}^2\right\} \eqsp,\nonumber\\
&= \sum_{\ell=1}^N\tilde{\omega}_{\mathsf{b},i}^\ell(a_i)\exp\left\{-\frac{1}{2}\normMat{\tilde{S}_{i+1}^\ell}{z_i}^2+z_i'\tilde{s}_{i+1}^\ell\right\}\eqsp,\label{eq:tN}
\end{align}
where
\begin{align*}
C_i(\tilde{a}^{\ell}_{i+1:n}) &\eqdef \exp \left\{\tilde{c}^{\ell}_{i+1}/2\right\}  \int_{z_{i+1}} \gamma^{}_{i+1}(\tilde{a}^{\ell}_{i+1}, z)\tilde{p}(y_{i+1:n}|\tilde{a}^{\ell}_{i+1:n},z) \rmd z\eqsp,\\
\tilde{\omega}_{\mathsf{b},i}^\ell(a_i) &= \tilde{p}_{i+1}(y_{i+1:n})C_i(\tilde{a}^{\ell}_{i+1:n})^{-1}Q(a_i, \tilde{a}^{\ell}_{i+1}) \tilde{\omega}^{\ell}_{i+1} |\barH_{\tilde{a}^{\ell}_{i+1}}|^{-1/2} |H_{\tilde{a}^{\ell}_{i+1}} \Delta^{\ell}_{i+1} H'_{\tilde{a}^{\ell}_{i+1}}|^{1/2}\\
&\times\exp\{-d'_{\tilde{a}^\ell_{i+1}}\barH_{\tilde{a}^\ell_{i+1}}^{-1}d_{\tilde{a}^\ell_{i+1}}/2\}\exp\{(\tilde{\nu}^{\ell}_{i+1} + \barH_{\tilde{a}^\ell_{i+1}}^{-1}d_{\tilde{a}^\ell_{i+1}})'H_{\tilde{a}^{\ell}_{i+1}}\Delta^{\ell}_{i+1}H'_{\tilde{a}^{\ell}_{i+1}}(\tilde{\nu}^{\ell}_{i+1} + \barH_{\tilde{a}^\ell_{i+1}}^{-1}d_{\tilde{a}^\ell_{i+1}})/2\} \eqsp,\\
(\tilde{S}_{i+1}^\ell)^{-1} &= T'_{\tilde{a}^{\ell}_{i+1}}\barH_{\tilde{a}^\ell_{i+1}}^{-1}(T_{\tilde{a}^{\ell}_{i+1}}-H_{\tilde{a}^{\ell}_{i+1}}\Delta^{\ell}_{i+1}H'_{\tilde{a}^{\ell}_{i+1}}\barH_{\tilde{a}^\ell_{i+1}}^{-1}T_{\tilde{a}^{\ell}_{i+1}})\eqsp,\\
\tilde{s}_{i+1}^\ell &=T'_{\tilde{a}^{\ell}_{i+1}}\barH_{\tilde{a}^\ell_{i+1}}^{-1}(H_{\tilde{a}^{\ell}_{i+1}}\Delta^{\ell}_{i+1}H'_{\tilde{a}^{\ell}_{i+1}}(\tilde{\nu}^{\ell}_{i+1} + \barH_{\tilde{a}^\ell_{i+1}}^{-1}d_{\tilde{a}^\ell_{i+1}})-d_{\tilde{a}^\ell_{i+1}})\eqsp.
\end{align*}
In the numerical experiments, $\gamma^{}_i(a_i,z_i)$ is set as a mixture of Gaussian distributions. As explained in Section~\ref{eq:rbtf}, the integral $\int_{z_{i+1}} \gamma^{}_{i+1}(\tilde{a}^{\ell}_{i+1}, z)\tilde{p}(y_{i+1:n}|\tilde{a}^{\ell}_{i+1:n},z) \rmd z$ may be computed explicitly, see Lemma~\ref{lem:integral:gammap}.

\section{Simulated data}
\label{sec:numerical:experiments}
This section highlights the improvements brought by  the additional Rao-Blackwellization steps for the two-filter and the FFBS approximations of the marginal smoothing distributions in the case where the number of states is $J=2$. The transition matrix $Q$ is such that the probability of switching from one regime to the other is small, as expected for the WTI crude oil data, see Section~\ref{sec:exp}. First, the algorithms are applied to a simple one-dimensional model with:
\[
\pi_1=\pi_2 = 0.5\\;\;d_1 = 0.5\;\; d_2 = 0;\;c_1= 0.1\;\; c_2 = 0 \eqsp,
\]
\[
Q = \begin{pmatrix}0.99 & 0.01\\ 0.03 & 0.97\end{pmatrix}\;\;T_1 = T_2 = 1\;\;\barH_1 = \barH_2= 0.1\eqsp,
\]
\[
B_1 = B_2 =1\;\;\barG_1 = 0.3\;\;\barG_2 = 0.1\eqsp.
\]
The original FFBS algorithm of \cite{lindsten:bunch:sarkka:schon:godsill:2015} and the FFBS algorithm with rejuvenation proposed in this paper are used with $N = \tilde{N} = 25$. For comparable computational costs, the two-filter method of \cite{briers:doucet:maskell:2010} and the method with rejuvenation are run with $N=100$.  The artificial distributions are chosen as $\gamma^{}_i (a_i, z_i) = p^N(a_{i},z_{i}|y_{1:i-1})$ where $p^N(a_{i},z_{i}|y_{1:i-1})$ is defined by \eqref{eq:rN}. All these algorithms are compared to the estimation obtained with the proposed FFBS algorithm with rejuvenation and $5000$ particles considered as a benchmark value. Figure~\ref{fig:ffbs:err} displays the mean estimation error over $100$ independent Monte Carlo runs. The estimation error is defined as the absolute difference between the benchmark value and the estimations given by all algorithms.
\begin{figure}
\centering
\includegraphics[scale=.3]{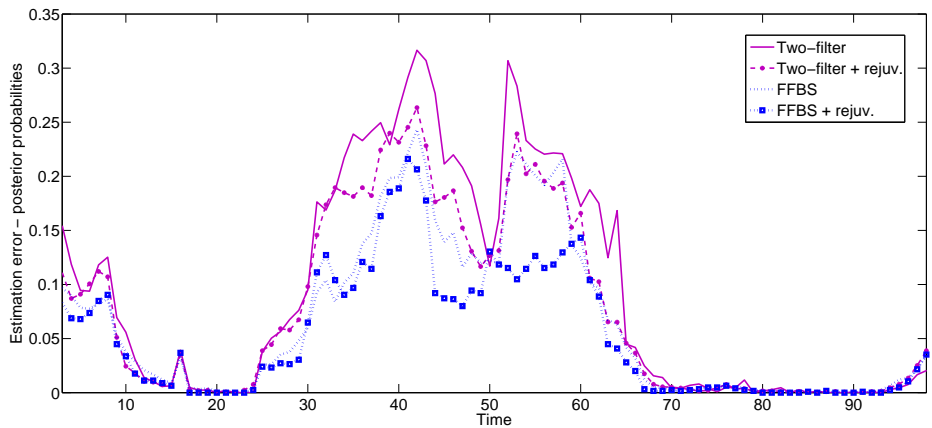}
\caption{Posterior probabilities estimation error for all algorithms.}
\label{fig:ffbs:err}
\end{figure}
In addition, Figure~\ref{fig:ffbs:var} displays the empirical variance of the estimation for each algorithm.
Figure~\ref{fig:ffbs:err} and  Figure~\ref{fig:ffbs:var} illustrate that in both cases the additional rejuvenation step improves the accuracy and the variability of SMC smoothers. In addition, even with a sharp choice for the artificial distributions $\gamma_i$, $1\le i\le n$,  FFBS based methods outperform two-filter based smoothers for this model. 
\begin{figure}
\centering
\includegraphics[scale=.3]{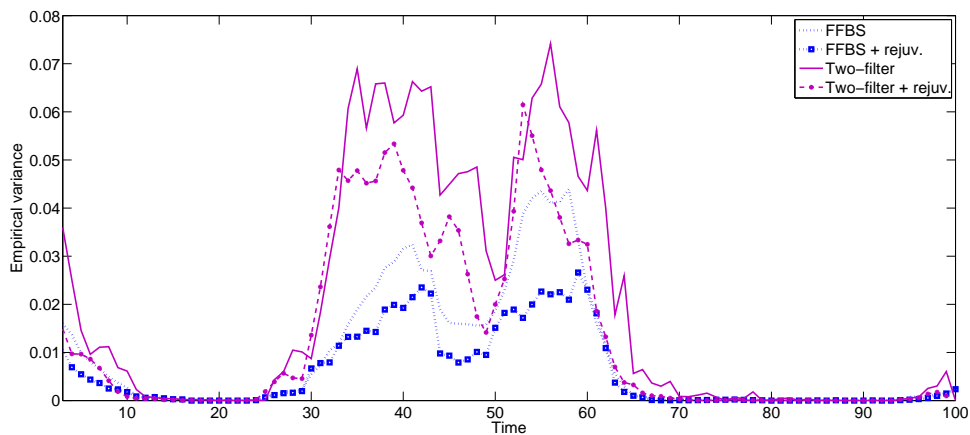}
\caption{Empirical variances of the estimation of $\mathbb{P}(a_k=1|Y_{1:n})$ for all algorithms.}
\label{fig:ffbs:var}
\end{figure}

\section{Application to CME crude oil (WTI)}
\label{sec:crudeoil}
\label{sec:exp}

\subsection{Model}
Modeling commodity prices is a crucial step to valuate contingent claims related to energy markets and to optimize storage or extraction strategies. In \cite{gibson:schwartz:1990,schwartz:1997}, the authors proposed a model where the spot price of a commodity$(S_t, t \geq 0)$ depends on a second factor $(\delta_t, t \geq 0)$, referred to as the instantaneous convenience yield. This factor plays the role of dividends in equity markets and models the benefit of holding the physical commodity or the storage and maintenance costs required to keep the commodity. In this model, this convenience yield is described as an Ornstein-Uhlenbeck process:
\begin{align*}
\rmd S_{t} & = (r-\delta_{t})S_{t}\rmd t+\sigma S_{t} \rmd W_{t}^{1} \eqsp,\\
\rmd \delta_{t} & = \kappa(\alpha-\delta_{t})\rmd t+\eta \rmd W_{t}^{2} \eqsp, \quad \rmd \langle W_t^1, W_t^2\rangle = \rho \rmd t\eqsp,
\end{align*}
where the parameter $\left(r, \sigma, \kappa, \alpha, \eta, \rho \right)$ are constant and $( (W_t^1,W_t^2), t \geq 0)$ are standard Brownian motions. This model appears to be too restrictive as energy markets are not likely to revert to a single equilibrium value. This assumption is relaxed using Markov switching models to allow several possible regimes for the spot price and the convenience yield. Following \cite{almansour:2016}, the spot price and convenience yield are described in this paper as:
\begin{align*}
\rmd S_{t} & = (r-\delta_{t})S_{t}\rmd t+\sigma_{a_t} S_{t}\rmd W_{t}^{1} \eqsp,\\
\rmd\delta_{t} & = \kappa (\alpha_{a_t} -\delta_{t})\rmd t+\eta_{a_t} \rmd W_{t}^{2} \eqsp, \quad \rmd\langle W_t^1, W_t^2\rangle = \rho_{a_t} \rmd t\eqsp,
\end{align*}
where $(a_t)_{t\ge 0}$ is a finite state space Markov process. This model allows to exhibit fundamental features of commodity future prices, which typically display different regimes of volatility and/or convenience yield. A two-regime model is already sufficient to produce stylized effects such as contango (increase of future prices) and backwardation (decrease of future prices).  Assuming that the switching rate between regimes is negligible compared to the inverse of the discretization period, the discretized version of the spot price and convenience yield $Z_i= (\ln S_i, \delta_i)$ (the sampling period is taken to be 1) is modeled as a CLGM. The explicit integration of this SDE detailed in Lemma~\ref{lem:integratedSDE} yields the following discrete time model for $(Z_i)_{i\ge 2}$:
\[
Z_i = d_{a_{i-1}} + T Z_{i-1} + H_{a_{i-1}} \varepsilon_i\eqsp,
\]
where (with $\barH_{a_{i-1}} \eqdef H_{a_{i-1}}H'_{a_{i-1}}$ and $\tau = t_i-t_{i-1}$):
\begin{align*}
d_{a_{i-1}} &\eqdef
\begin{pmatrix} \left[\mu- \alpha_{a_{i-1}} - \sigma^2_{a_{i-1}}/2 \right] \tau + \alpha_{a_{i-1}}[1-\rme^{-\kappa \tau}]/\kappa\\
\alpha_{a_{i-1}} [1-\rme^{-\kappa \tau}] \end{pmatrix}\eqsp,\\
T &\eqdef
\begin{pmatrix} 1 & -[1-\rme^{-\kappa \tau}]/\kappa \\ 0 & \rme^{-\kappa \tau}  \end{pmatrix} \eqsp,\\
\barH_{a_{i-1}}(1,1) &= \sigma^2_{a_{i-1}} \tau + \eta^2_{a_{i-1}}\left\{\tau  + (1-\rme^{-2\kappa \tau})/(2\kappa) - 2(1-\rme^{-\kappa \tau})/\kappa\right\}/\kappa^2 \\
&\hspace{4.4cm} - 2\rho_{a_{i-1}} \eta_{a_{i-1}} \sigma_{a_{i-1}}\left\{ \tau t_i - (1-\rme^{-\kappa \tau})/\kappa \right\}/\kappa\eqsp, \\
\barH_{a_{i-1}}(1,2) & = \left(\rho_{a_{i-1}} \eta_{a_{i-1}} \sigma_{a_{i-1}}-\eta^2_{a_{i-1}}/\kappa \right)\left(1-\rme^{-\kappa \tau} \right)/\kappa + \eta^2_{a_{i-1}}\left(1-\rme^{-2\kappa \tau} \right)/(2\kappa^2)\eqsp,\\
\barH_{a_{i-1}}(2,1) &= \barH_{a_{i-1}}(1,2)\eqsp, \quad \barH_{a_{i-1}}(2,2) = \eta^2_{a_{i-1}}\left( 1-\rme^{-2\kappa \tau} \right)/(2\kappa)\eqsp.
\end{align*}
The observations are Wednesday future contracts of the West Texas Intermediate crude oil (WTI) traded in the Chicago Mercantile Exchange (CME) from $11$ January 1995 to $13$ November 2013. The contracts are numbered $F_1, F_2, \ldots, F_{36}$ where $F_1$ (or front month) is the earliest delivery future contract, $F_2$ is the second earliest delivery future contract and so on. Among these $36$ contracts, the four future contracts: $F_1, F_4, F_6, F_{13}$ are used since their trading volumes and their impacts on the Term Structures are the most important ($F_1$ is the most liquid contract, $F_{13}$ characterizes the gap between prices over a one year period, $F_4$ and $F_6$ are intermediate future contracts that are mostly traded). As in \cite{almansour:2016}, we consider that each  future contract has a fixed time to maturity: $F_1, F_4, F_6, F_{13}$ have time to  maturity $4$, $16$, $26$, $56$ weeks.  Our time series contains $n =975$ weekly data with $534$ in backwardation and $441$ in contango (the backwardation effect is more frequent with crude oil data). At each time $t_i = i\tau$, with $\tau = 0.0192$, the observations of the $\dimy = 4$ future prices are $Y_i \eqdef (\ln(F^{(market)}_{i\tau t,m_1}), \ldots, \ln(F^{(market)}_{i\tau t,m_{\dimy}}))'$, where $F_{t_i,m}$ is the future price at $t_i$ for a maturity  $m$ weeks.
A closed form solution for $F_{t_i,m}$  may be written:
\[
F_{t_i,m} \eqdef \exp\left( \mathsf{A}_{m}(a_{i}) + \mathsf{B}_{m} Z_{i}\right)\eqsp,
\]
where $\mathsf{B}_0	= \begin{pmatrix}1 & 0 \end{pmatrix}$ and  $\mathsf{B}_m    = \mathsf{B}_{m-1} T$ so that  $ \mathsf{B}_m= \begin{pmatrix} 1 & -\left(1-\rme^{-\kappa m \tau t}\right)/\kappa \end{pmatrix}$ and
for all $1\le j\le J$, $\mathsf{A}_0(j) = 0$,  and
\begin{equation*}
\mathsf{A}_m(j) = \ln \left(\sum^J_{k=1} Q(j,k) \exp(\mathsf{A}_{m-1}(k)) \right) + \mathsf{B}_{m-1} d_j + \frac{1}{2} \mathsf{B}_{m-1} \barH_j B'_{m-1}\eqsp.
\end{equation*}
Therefore, the observations of the  logfuture prices are given, for all $1\le i\le n$, by:
\[
Y_i = c_{a_i} + B Z_i + G \eta_i\eqsp,
\]
where $\eta_i$ is a standard multivariate Gaussian random variable and:
\[
c_{j}' = [ \mathsf{A}_{m_1}(j), \dots, \mathsf{A}_{m_{\dimy}}(j) ]\eqsp,
\;\;\;
B' = [\mathsf{B}_{m_1}',  \dots, \mathsf{B}_{m_{\dimy}}']\eqsp, \;\;\;
G = \mathrm{diag}(g_1, \dots , g_{d} )\eqsp.
\]
The model depends of the parameters:
\[
\theta \eqdef \{\pi, Q, \mu_1, \Sigma_1, \kappa, (\alpha_j)_{1\leq j\leq J}, (\sigma_j)_{1\leq j\leq J}, (\eta_j)_{1\leq j\leq J}, (\rho_j)_{1\leq j\leq J}, (g_{\ell})_{1\leq \ell\leq d} \}\eqsp .
\]
The aim of this section is to estimate $\theta$ and the posterior probabilities $\mathbb{P}(a_k=j|Y_{1:n})$, $1\le k \le n$, $1\le j \le J$. Given the observations $Y_{1:n}$, the EM algorithm introduced in \cite{dempster:laird:rubin:1977} maximizes the incomplete data log-likelihood $\theta\mapsto \ell_{\theta}^{n}$ defined by
\begin{equation*}
\ell_{\theta}^{n}(Y_{1:n}) \eqdef \log\left(\sum_{a_1=1}^J\ldots\sum_{a_n=1}^J\int p_{\theta}(a_{1:n},z_{1:n},Y_{1:n})\,\rmd z_{1:n}\right)\eqsp,
\end{equation*}
where the complete data likelihood $p_{\theta}$ is given by
\begin{equation*}
p_{\theta}(a_{1:n},z_{1:n},Y_{1:n}) \eqdef \pi(a_1)\phi_{\mu_1,\Sigma_1}(z_1)g_{\theta}(a_1,z_1;Y_1)\prod^{n}_{i=2}Q(a_{i-1},a_i)m_{\theta}\left(a_{i},z_{i-1};z_i\right)g_{\theta}(a_i,,z_i;Y_i)\eqsp.
\end{equation*}
Denote by $\mathbb{E}_{\theta}\left[\cdot\middle|Y_{1:n}\right]$  the conditional expectation given $Y_{1:n}$ when the parameter value is set to $\theta$. The EM algorithm iteratively builds a sequence $\{\theta_{p}\}_{p\ge 0}$ of parameter estimates following the two steps:
\begin{enumerate}
	\item {\bf E-step}: compute $\theta \mapsto Q(\theta,\theta_{p})\eqdef \mathbb{E}_{\theta_p}\left[\log p_{\theta}(a_{1:n},Z_{1:n},Y_{1:n})\middle|Y_{1:n}\right]$ ;
	\item {\bf M-step}: choose $\theta_{p+1}$ as a maximizer of $\theta \mapsto Q(\theta,\theta_{p})$.
\end{enumerate}
All the conditional expectations involved in $Q(\theta,\theta_p)$ are approximated using our two-filter algorithm with rejuvenation to define the SMC approximation  $\theta\mapsto Q^N(\theta,\theta_p)$ of $\theta\mapsto Q(\theta,\theta_p)$. As the function $\theta\mapsto Q^N(\theta,\theta_p)$ cannot be maximized analytically, the M-step is performed numerically using the Covariance Matrix Adaptation Evolution Strategy (CMA-ES) introduced in \cite{hansen:ostermeier:2001}. This derivative-free optimization procedure is known to perform well in complex multimodal optimization settings, see e.g. \cite{hansen:kern:2004}.

\subsection{Numerical results}
The initial transition probability in CMA-ES is chosen as $Q(1,1) = 0.98, \; Q(2,2) = 0.97$ and $\pi_1 = \pi_2 = 0.5$ where the number $1$ represents the backwardation regime and $2$ represents the contango regime. The other parameters are initialized as shown in Table~\ref{fig:CMEinitialestimates}.
\begin{table}[h!]
	\centering
	\begin{tabular}{||c|c|c|c|c|c|c|c|c|c|c|c|c||}
		$\kappa$ & $\alpha_1$ & $\alpha_2$ & $\sigma_1$ & $\sigma_2$ & $\eta_1$ & $\eta_2$ & $\rho_1$ & $\rho_2$ & $g_1$ & $g_2$ & $g_3$ & $g_4$ \\
		\hline
		5.0 & 0.1 & -0.05 & 0.4 & 0.4 & 0.5 & 0.5 & 0.75 & 0.65 & 0.1 & 0.1 & 0.1 & 0.1
	\end{tabular}
	\caption{Initial values for the EM algorithm.}
	\label{fig:CMEinitialestimates}
\end{table}
The number of particles is set to $N = 100$, $\tau  = 1/52$. 
The interest rate is $r = 0.0296$ as in \cite{almansour:2016}. The initial guess for the mean and variance of the initial state are
\[
\mu_1 = \begin{pmatrix} \ln F^{(market)}_{1,4} & r-\cfrac{\ln F^{(market)}_{1,16}-\ln F^{(market)}_{1,4}}{(16-4)\tau} \end{pmatrix}\quad \mbox{and}\quad \Sigma_1 = \begin{pmatrix} 0.05 & 0 \\ 0 & 0.05 \end{pmatrix}\eqsp.
\]
The CMA-ES algorithm is used with an initial standard deviation for the parameters $\sigma_{cmaes} = 0.005$, a number of selected search points $\mu_{cmaes} = 20$ and a population size $\lambda_{cmaes} = 100$. The algorithm is stopped after $10000$ iterations.

In Gibson-Schwartz model \cite{gibson:schwartz:1990}, a stronger backwardation effect implies a greater value for $\alpha$ for the same values of the other parameters. For the CME WTI Crude Oil, backwardation effect is more frequent than contango effect so that $\alpha_1$ should be greater than $\alpha_2$. Therefore, this condition is imposed for all simulations in the CMA-ES algorithm. The results after $2500$ iterations of the EM algorithm are given in Table~\ref{fig:CMEfinalestimates}. The estimated values and standard deviations are obtained with $50$ independent runs of the algorithm.
\begin{table}[h!]
	\centering
	\begin{tabular}{||c|c|c|c|c|c|c|c|c||c|clc|c|c||}
	\hline
		Parameter&  $\kappa$ & $\sigma_1$ & $\sigma_2$ & $\eta_1$ & $\eta_2$ & $\rho_1$ & $\rho_2$  \\
		\hline
		Value & 2.6378 &0.3733 & 0.3485 & 0.5892 & 0.3814 & 0.8709 & 0.6761   \\
		\hline
		Std. Dev & 0.1999 & 0.005438 & 0.002884 & 0.0483 & 0.0349 & 0.0064 & 0.0052 \\  
		\hline
	\end{tabular}
	
	\vspace{.2cm}
	
	\begin{tabular}{||c|c|c|c|c|c|c|c|c|c|c|||}
	\hline
		Parameter & $\alpha_1$ & $\alpha_2$  & $g_1$ & $g_2$ & $g_3$ & $g_4$  & $Q(1,1)$ & $Q(2,2)$\\
		\hline
		Value & 0.0889 & -0.0281  & 2.3e-2 & 1.0e-4 & 3.0e-4 & 2.3e-2 & 0.9917 & 0.9880 \\
		\hline
		Std. Dev & 0.004248 & 0.00149 & 1.9e-4 & 2.6e-4 & 2.3e-4 & 2.1e-04 &  6.7e-4 & 9.6e-4\\
		\hline
	\end{tabular}
	\caption{Final estimates after 2500 iterations.}
	\label{fig:CMEfinalestimates}
\end{table}
As expected, we obtain $\sigma_1 \geq \sigma_2$, $\alpha_1 \geq \alpha_2$, $\eta_1 \geq \eta_2$ and $\rho_1 \geq \rho_2$ at convergence of the EM algorithm. Moreover, $Q(1,1) > Q(2,2)$ corresponds to the prediction that we did from the data description. The fact that $\sigma_1 \geq \sigma_2$ and $\alpha_1 \geq \alpha_2$ indicates the first regime (backwardation) characterized by a higher value in both volatility and equilibrium level of convenience yield, and the second regime (contango) characterized by a lower value in both volatility and equilibrium convenience yield level. This in accordance with the theory of storage that the volatility of the commodity spot price is high when the inventory is low, and the convenience yield is all the higher as inventory is low.

Figure~\ref{fig:LogPriceTermStructure:data} compares the evolution of future $1M$ (the nearest contracts) to the \emph{term structure} observed from CME WTI crude oil, defined as the difference of future $13M$ and future $1M$ (to avoid seasonality). The figure shows that it is not necessary to have an inverse relationship between the  price of the nearest contract and the term structure. But when a significant drop in the price  of the nearest contract occurs, the term structure increases (i.e. in contango).
\begin{figure}
\centering
\includegraphics[scale=.6,angle=90]{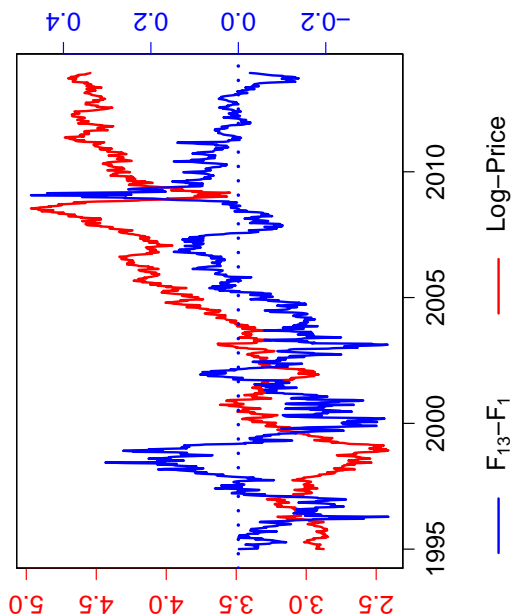}
\caption{Log-price (red line) and slope of future curves (blue line).}
\label{fig:LogPriceTermStructure:data}
\end{figure}

The correlation between the spot price and the convenience yield is positive and high in both two regimes. This is an accordance to what as been observed in most commodity market, see \cite{gibson:schwartz:1990}.  The slope of future curve decreases in function of maturity. 

Figure~\ref{fig:posterior:data} and~\ref{fig:posterior:data:logprice} display the the estimated posterior probabilities of the regimes and the observed future slope. When the future curve is in backwardation (resp. contango), the model is expected to be in the first regime (resp. second regime), except for the period where the slope of the future curve is too small and in the period from December $2008$ to April $2009$ (beginning of the crisis).

\begin{figure}
	\centering
	\includegraphics[scale=.75,angle=90]{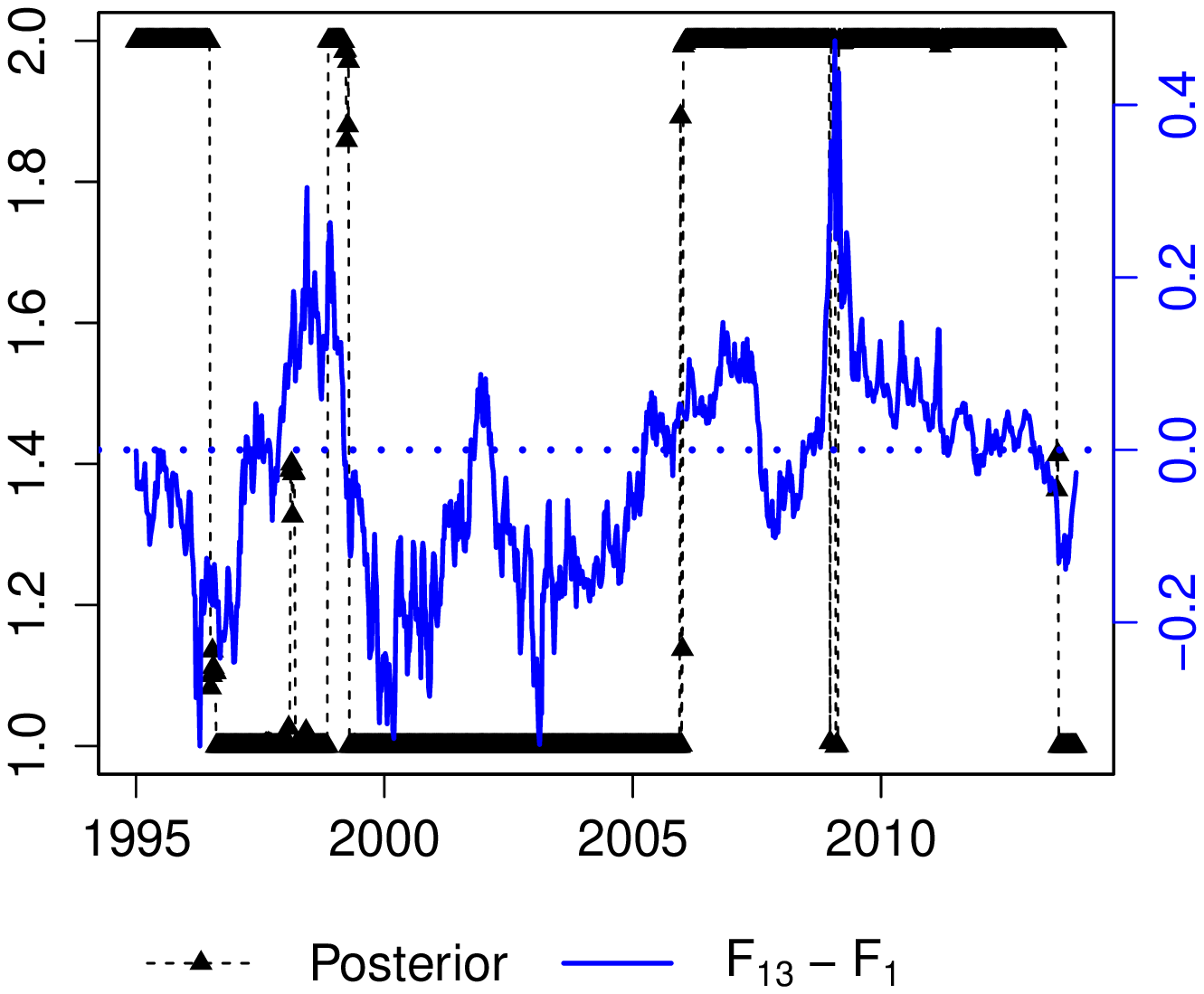}
	\caption{Posterior probability (triangle black) and slope of future curves (blue line).}
	\label{fig:posterior:data}
\end{figure}

\begin{figure}
	\centering
	\includegraphics[scale=.75]{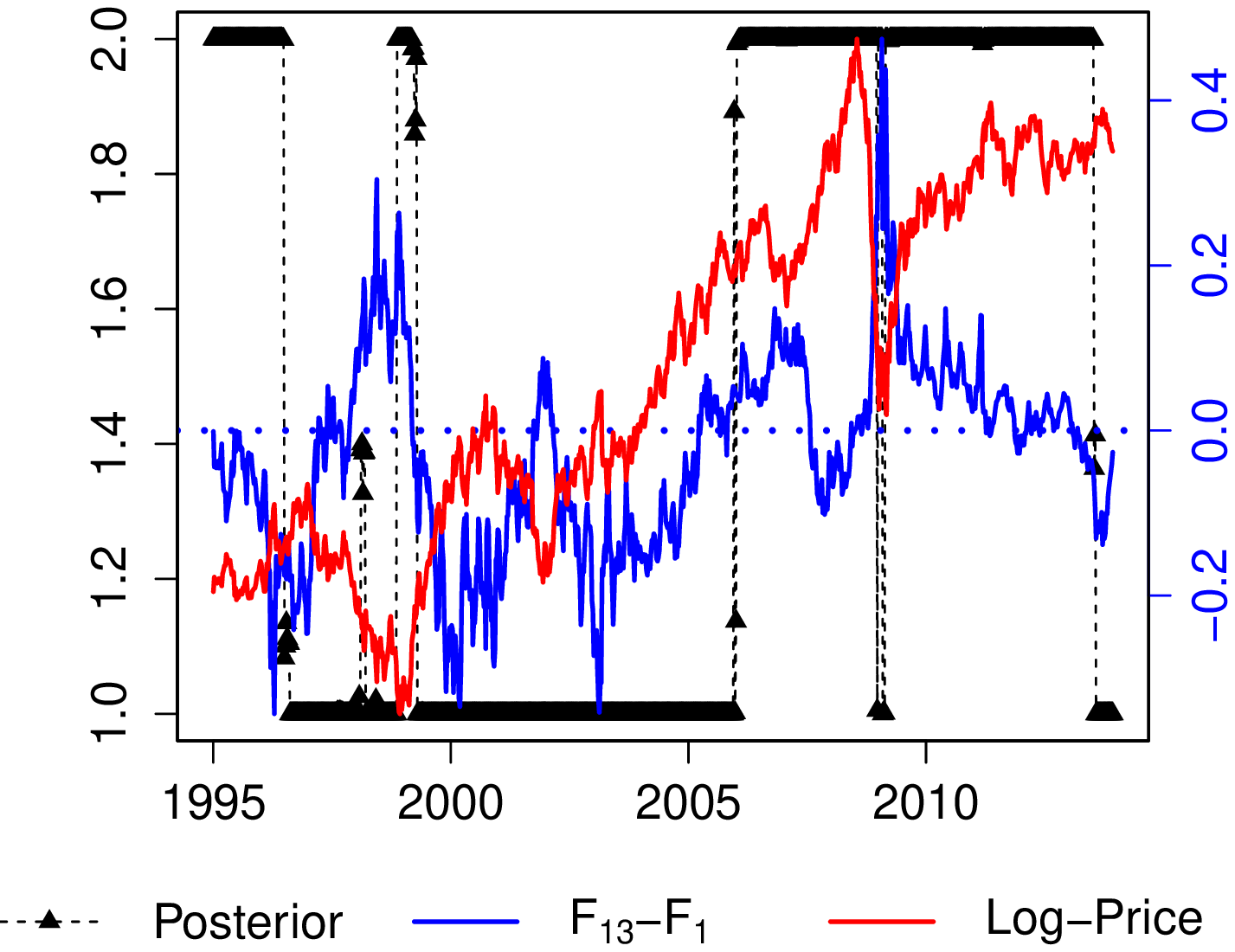}
	\caption{Posterior probability (triangle black) and slope of future curves (blue line).}
	\label{fig:posterior:data:logprice}
\end{figure}

\section{Conclusions}
This paper presents Rao-Blackwellized Sequential Monte Carlo methods to approximate smoothing distributions in conditionally linear and Gaussian state spaces in a common unifying framework. It also provides different techniques that could be used in the forward filtering pass to improve significantly the usual mixture Kalman filter. The filtering distributions are approximated at  each time step by considering all possible offsprings of all ancestral trajectories before discarding degenerated paths instead of resampling the ancestral paths before propagating them at the next time step. The paper investigates the benefit of additional Rao-Blackwellization steps to sample new regimes at each time step conditional on the forward and backward particles. This rejuvenation step uses explicit integration of the hidden linear states before merging the forward and backward filters for two-filter based algorithms or before sampling new states backward in time for FFBS based methods.  The paper displays some Monte Carlo experiments with simulated data to illustrate that this additional rejuvenation step improves the performance of the smoothing algorithms with no substantial additional computational costs. They are also applied to commodity markets using WTI crude oil data.

\appendix

\section{Technical lemmas}
\label{sec:app}
Lemmas~\ref{eq:technical:twofilter}, \ref{lem:py} and \ref{lem:integral:gammap} are close to \cite[Proposition 5, Proposition 6]{briers:doucet:maskell:2010}. The proofs are detailed in this appendix for completeness.
\begin{proof}[Proof of Lemma~\ref{eq:technical:twofilter}]
For all $1\le i \le n-1$,
\begin{align*}
p_{}(y_{i:n}|a_i,z_i) &= \sum_{a_{i+1:n}}\int p_{}(y_{i:n},a_{i+1:n},z_{i+1:n}|a_i,z_i)\rmd z_{i+1:n}\eqsp,\\
&=  \sum_{a_{i+1:n}}\int p_{}(a_{i+1:n},z_{i+1:n}|a_i,z_i)p_{}(y_{i:n}|a_{i:n},z_{i:n})\rmd z_{i+1:n}\eqsp,\\
&=  \frac{\tilde{p}_{i}(y_{i:n})}{\gamma^{}_{i}(a_i,z_i)}\sum_{a_{i+1:n}}\int \frac{\gamma^{}_{i}(a_i,z_i)}{\tilde{p}_{i}(y_{i:n})}p_{}(a_{i+1:n},z_{i+1:n}|a_i,z_i)p_{}(y_{i:n}|a_{i:n},z_{i:n})\rmd z_{i+1:n}\eqsp,\\
&=  \frac{\tilde{p}_{i}(y_{i:n})}{\gamma^{}_{i}(a_i,z_i)}\sum_{a_{i+1:n}}\int \tilde{p}_{i}(a_{i:n},z_{i:n}|y_{i:n})\rmd z_{i+1:n}\eqsp,\\
&=  \frac{\tilde{p}_{i}(y_{i:n})}{\gamma^{}_{i}(a_i,z_i)}\tilde{p}_{i}(a_{i},z_{i}|y_{i:n})\eqsp,
\end{align*}
which concludes the proof of \eqref{eq:p:y1:n}. To prove \eqref{eq:p:ai:zi} write,
\begin{align*}
\tilde{p}_{i}(a_{i:n},z_{i}|y_{i:n}) &= \frac{\gamma^{}_{i}(a_i,z_i)}{\tilde{p}_{i}(y_{i:n})}\int p_{}(y_{i:n}|a_{i:n},z_{i:n})p_{}(a_{i+1:n},z_{i+1:n}|a_{i},z_{i})\rmd z_{i+1:n}\eqsp,\\
&= \frac{\gamma^{}_{i}(a_i,z_i)}{\tilde{p}_{i}(y_{i:n})}\int \frac{p_{}(y_{i:n}|a_{i:n},z_{i})p_{}(z_{i+1:n}|y_{i:n},a_{i:n},z_i)}{p_{}(z_{i+1:n}|a_{i:n},z_i)}p_{}(a_{i+1:n},z_{i+1:n}|a_{i},z_{i})\rmd z_{i+1:n}\eqsp,\\
&=\frac{\gamma^{}_{i}(a_i,z_i)}{\tilde{p}_{i}(y_{i:n})}p_{}(y_{i:n}|a_{i:n},z_{i})p_{}(a_{i+1:n}|a_i)\eqsp.
\end{align*}
Therefore,
\[
\tilde{p}_{i}(a_{i},z_{i}|y_{i:n}) = \frac{\gamma^{}_{i}(a_i,z_i)}{\tilde{p}_{i}(y_{i:n})}\sum_{a_{i+1:n}}p_{}(y_{i:n}|a_{i:n},z_{i})p_{}(a_{i+1:n}|a_i)
\]
and the proof is completed upon noting that
\[
\tilde{p}_{i}(a_{i:n}|y_{i:n}) = \frac{p_{}(a_{i+1:n}|a_i)}{\tilde{p}_{i}(y_{i:n})} \int \gamma^{}_{i}(a_i,z)p_{}(y_{i:n}|a_{i:n},z)\rmd z\eqsp.
\]
\end{proof}

\begin{lemma}
\label{lem:py}
For all $1\le i \le n$,
\begin{equation}
\label{eq:py}
p_{}(y_{i:n}|a_{i:n},z_i) =  \exp\left\{-\frac{1}{2}\tilde{c}_{i}(a_{i:n}) - \frac{1}{2}\normMat{\tilde{P}_{i}(a_{i:n})}{z_i}^2 + z'_i\tilde{\nu}_{i}(a_{i:n})\right\}\eqsp,
\end{equation}
where
\begin{align}
\tilde{c}_{n}(a_{n}) &= \dimy\log(2\pi) + \log \left|\barG_{a_n}\right| + \normMat{\barG_{a_n}}{y_n-c_{a_n}}^2\eqsp,\label{eq:cn}\\
\tilde{P}_{n}^{-1}(a_{n}) &=B'_{a_n}\barG_{a_n}^{-1}B_{a_n}\eqsp,\label{eq:Pn}\\
\tilde{\nu}_{n}(a_{n})&=B'_{a_n}\barG_{a_n}^{-1}(y_n-c_{a_n})\label{eq:nun}
\end{align}
and, for all $1\le i\le n-1$,
\begin{align}
\tilde{c}_{i}(a_{i:n}) &= \tilde{c}_{i|i+1}(a_{i+1:n}) + \dimy\log(2\pi)  + \log|\barG_{a_i}| + \normMat{\barG_{a_i}}{y_i-c_{a_i}}^2\eqsp,\label{eq:ck}\\
\tilde{P}_{i}^{-1}(a_{i:n}) &= \tilde{P}_{i|i+1}^{-1}(a_{i+1:n}) + B'_{a_i}\barG_{a_i}^{-1}B_{a_i}\eqsp,\label{eq:Pk}\\
\tilde{\nu}_{i}(a_{i:n})&=\tilde{\nu}_{i|i+1}(a_{i+1:n}) + B'_{a_i}\barG_{a_i}^{-1}(y_{i}-c_{a_{i}})\eqsp,\label{eq:nuk}
\end{align}
with
\begin{align*}
\Delta_{i+1}(a_{i+1:n}) &= \left(I_{\dimz} + H'_{a_{i+1}}\tilde{P}_{i+1}^{-1}(a_{i+1:n})H_{a_{i+1}}\right)^{-1}\eqsp,\\
\tilde{r}_{i|i+1}(a_{i+1:n})&= \tilde{\nu}_{i+1}(a_{i+1:n}) + \barH_{a_{i+1}}^{-1}d_{a_{i+1}}\eqsp,\\
\tilde{c}_{i|i+1}(a_{i+1:n}) &= \tilde{c}_{i+1}(a_{i+1:n}) + \log|\barH_{a_{i+1}}| + d'_{a_{i+1}}\barH_{a_{i+1}}^{-1}d_{a_{i+1}}-\log|H_{a_{i+1}}\Delta_{i}(a_{i+1:n})H'_{a_{i+1}}|\\
&\hspace{4.0cm} - \tilde{r}'_{i|i+1}(a_{i+1:n})H_{a_{i+1}}\Delta_{i}(a_{i+1:n})H'_{a_{i+1}}\tilde{r}_{i|i+1}(a_{i+1:n})\eqsp,\\
\tilde{P}_{i|i+1}^{-1}(a_{i+1:n}) &=T'_{a_{i+1}}\left(I_{\dimz}-\barH_{a_{i+1}}^{-1}H_{a_{i+1}}\Delta_{i}(a_{i+1:n})H'_{a_{i+1}}\right)\barH_{a_{i+1}}^{-1}T_{a_{i+1}}\eqsp,\\
\tilde{\nu}_{i|i+1}(a_{i+1:n})&=T'_{a_{i+1}}\barH_{a_{i+1}}^{-1}\left[-d_{a_{i+1}}+H_{a_{i+1}}\Delta_{i}(a_{i+1:n})H'_{a_{i+1}}\left(\tilde{\nu}_{i+1}(a_{i+1:n})+\barH_{a_{i+1}}^{-1}d_{a_{i+1}}\right)\right]\eqsp.
\end{align*}
\end{lemma}

\begin{proof}
The result is proved by backward induction. \eqref{eq:cn}, \eqref{eq:Pn} and \eqref{eq:nun} follow directly from \eqref{eq:model:obs}. Assume that for a given $1\le i\le n-1$, $p(y_{i+1:n}|a_{i+1:n},z_{i+1})$ is given by \eqref{eq:py}. Write
\[
p(y_{i:n}|a_{i:n},z_i) = \int m(a_{i+1},z_i;z_{i+1})g(a_i,z_i;y_i)p(y_{i+1:n}|a_{i+1:n},z_{i+1})\rmd z_{i+1}\eqsp,
\]
with
\begin{align*}
m(a_{i+1},z_i;z_{i+1}) &= \exp\left\{-\frac{\dimz}{2}\log(2\pi)-\frac{1}{2}\log|\barH_{a_{i+1}}|-\frac{1}{2}\normMat{\barH_{a_{i+1}}}{z_{i+1}-d_{a_{i+1}}-T_{a_{i+1}}z_i}^2\right\}\eqsp,\\
g(a_i,z_i;y_i)& = \exp\left\{-\frac{\dimy}{2}\log(2\pi)-\frac{1}{2}\log|\barG_{a_i}|-\frac{1}{2}\normMat{\barG_{a_i}}{y_{i}-c_{a_{i}}-B_{a_i}z_i}^2\right\}\eqsp,\\
p(y_{i+1:n}|a_{i+1:n},z_{i+1})& = \exp\left\{-\frac{1}{2}c_{i+1}(a_{i+1:n}) - \frac{1}{2}\normMat{\tilde{P}_{i+1}(a_{i+1:n})}{z_{i+1}}^2 + z'_{i+1}\tilde{\nu}_{i+1}(a_{i+1:n})\right\}\eqsp.
\end{align*}
Let $\Delta_{i+1}$ and  $\delta_{i+1}$ be given by:
\begin{align*}
\Delta_{i+1}(a_{i+1:n}) &\eqdef \left(I_{\dimz} + H'_{a_{i+1}}\tilde{P}_{i+1}^{-1}(a_{i+1:n})H_{a_{i+1}}\right)^{-1}\eqsp,\\
\delta_{i+1}(a_{i+1:n}) &\eqdef \nu_{i+1}(a_{i+1:n}) + \barH_{a_{i+1}}^{-1}(d_{a_{i+1}}+T_{a_{i+1}}z_i)\eqsp.
\end{align*}
Then, $\barH_{a_{i+1}}^{-1} + \tilde{P}_{i+1}^{-1}(a_{i+1:n}) = \left(H_{a_{i+1}}\Delta_{i+1}(a_{i+1:n})H'_{a_{i+1}}\right)^{-1}$ and \eqref{eq:ck}, \eqref{eq:Pk} and \eqref{eq:nuk} follows from
\begin{multline*}
\int \exp\left\{-\frac{1}{2}\normMat{H_{a_{i+1}}\Delta_{i+1}(a_{i+1:n})H'_{a_{i+1}}}{z_{i+1}}^2 + z'_{i+1}\delta_{i+1}(a_{i+1:n})\right\}\rmd z_{i+1} \\
=\exp\left\{\frac{1}{2}\log(2\pi) + \frac{1}{2}\log|H_{a_{i+1}}\Delta_{i+1}(a_{i+1:n})H'_{a_{i+1}}|\right\}\\
\times\exp\left\{\frac{1}{2}\delta'_{i+1}(a_{i+1:n})'H_{a_{i+1}}\Delta_{i+1}(a_{i+1:n})H'_{a_{i+1}}\delta_{i+1}(a_{i+1:n})\right\}\eqsp.
\end{multline*}
\end{proof}

\begin{lemma}
\label{lem:integral:gammap}
For all $1\le i\le n$,
\begin{multline*}
\int \phi_{\mu_i,\Sigma_i}(z_i)p_{}(y_{i:n}|a_{i:n},z_i)\rmd z_{i} = \exp\left\{-\frac{1}{2}\log |\Sigma_i| - \frac{1}{2}\mu_i'\Sigma^{-1}_i\mu_i\right\}\\
\times\exp\left\{-\frac{1}{2}\tilde{c}_{i}(a_{i:n})+\frac{1}{2}\log|\tilde{\Omega}_{i}(a_{i:n})| +\frac{1}{2}\tilde{z}'_i(a_{i:n})\tilde{\Omega}_{i}(a_{i:n})\tilde{z}_i(a_{i:n})\right\}\eqsp,
\end{multline*}
where $\phi_{\mu,\Sigma}$ is the probability density function of a $\dimz$ dimensional Gaussian random variable with mean $\mu$ and variance matrix $\Sigma$ and
\[
\tilde{\Omega}_{i}(a_{i:n})\eqdef \left(\Sigma_i^{-1} + \tilde{P}_{i}^{-1}(a_{i:n})\right)^{-1}\quad\mbox{and}\quad
\tilde{z}_i(a_{i:n})\eqdef \Sigma_i^{-1}\mu_i+\tilde{\nu}_{i}(a_{i:n})
\]
and where $c_{i}$, $\tilde{P}_{i}$ and $\nu_{i}$ are given in Lemma~\ref{lem:py}.
\end{lemma}
\begin{proof}
By Lemma~\ref{lem:py},
\begin{multline*}
\phi_{\mu_i,\Sigma_i}(z_i)p(y_{i:n}|a_{i:n},z_i) = \exp\left\{-\frac{\dimz}{2}\log(2\pi) -\frac{1}{2}\log|\Sigma_i|-\frac{1}{2}\normMat{\Sigma_i}{z_i-\mu_i}^2\right\}\\
\times\exp\left\{-\frac{1}{2}c_{i}(a_{i:n})-\frac{1}{2}\normMat{\tilde{P}_{i}(a_{i:n})}{z_i}^2+z'_i\nu_{i}(a_{i:n})\right\}\eqsp.
\end{multline*}
The proof is completed noting that
\begin{multline*}
\int \exp\left\{-\frac{1}{2}z'_i\tilde{\Omega}_{i}^{-1}(a_{i:n})z_i + z'_i(\Sigma_i^{-1}\mu_i+\nu_{i}(a_{i:n}))\right\}\rmd z_{i}\\
= \exp\left\{\frac{\dimz}{2}\log(2\pi) + \frac{1}{2}\log|\tilde{\Omega}_{i}(a_{i:n})| +\frac{1}{2}\left[\Sigma_i^{-1}\mu_i+\nu_{i}(a_{i:n})\right]'\tilde{\Omega}_{i}(a_{i:n})\left[\Sigma_i^{-1}\mu_i+\nu_{i}(a_{i:n})\right]\right\}\eqsp.
\end{multline*}\end{proof}

\begin{lemma}
\label{lem:integratedSDE}
Let $(X_t,\delta_t)_{t\ge0}$ be solutions to the following SDE:
\begin{align*}
\rmd X_t &= \left(\mu - \delta_t- \sigma^2/2 \right) dt + \sigma d W^1_t \eqsp,\\
\rmd \delta_t &= \kappa\left( \alpha - \delta_t\right)dt + \eta dW^2_t \eqsp,
\end{align*}
$(W_t^1)_{t\ge 0}$ and $(W_t^2)_{t\ge 0}$ are standard Brownian motions such that $\rmd \langle W_t^1,W_t^2 \rangle = \rho\rmd t$.
Then, for all $t\ge 0$ and $h>0$,
\[
\begin{pmatrix} X_{t+h} \\ \delta_{t+h}\end{pmatrix} = d_h + T_h\begin{pmatrix} X_{t} \\ \delta_{t}\end{pmatrix} + H_h\varepsilon\eqsp,
\]
where $\varepsilon$ is a standard 2-dimensional Gaussian random variable and (with $\barH_h\eqdef H_h'H_h$),
\[
d \eqdef
\begin{pmatrix} \left[\mu- \alpha - \sigma^2/2 \right] h + \alpha[1-e^{-\kappa h}]/\kappa\\
\alpha [1-e^{-\kappa h}] \end{pmatrix}\eqsp, \;\; T_h \eqdef
\begin{pmatrix} 1 & -[1-e^{-\kappa h}]/\kappa \\ 0 & e^{-\kappa h} \end{pmatrix} \eqsp,
\]
\begin{align*}
\barH_h(1,1) & \eqdef \sigma^2 h+ \eta^2\left\{h + (1-e^{-2\kappa h})/(2\kappa) - 2(1-e^{-\kappa h})/\kappa  \right\}/\kappa^2 \\
&\hspace{7cm} -2\rho\eta\sigma\left\{h - (1-e^{-\kappa h})/\kappa \right\}/\kappa\eqsp, \\
\barH_h(1,2) & \eqdef \left(\rho \eta \sigma-\eta^2/\kappa\right)\left(1-e^{-\kappa h} \right)/\kappa + \eta^2\left(1-e^{-2\kappa h} \right)/(2\kappa^2)\eqsp,\\
\barH_h(2,1) & \eqdef \barH_h(1,2)\eqsp,\\
\barH_h(2,2) & \eqdef \eta^2\left( 1-e^{-2\kappa h} \right)/(2\kappa)\eqsp.
\end{align*}
\end{lemma}

\begin{proof}
For all $t\ge 0$,
\[
X_t = X_0 + (\mu-\sigma^2/2)t - \int_0^t\delta_s\rmd s +\sigma W_t^1
\]
and, as $(\delta_t)_{0\le t\le T}$ is an Ornstein-Uhlenbeck process,
\[
\delta_t = \delta_0\rme^{-\kappa t} + \alpha(1-\rme^{-\kappa t}) + \int_0^t \eta \rme^{\kappa(s-t)}\rmd W^2_s\eqsp.
\]
Then,
\begin{align*}
\int_0^t\delta_s\rmd s &= (\delta_0-\alpha)(1-\rme^{-\kappa t})/\kappa + \alpha t + \eta\int_0^t\int_0^s\rme^{\kappa(u-s)}\rmd W^2_u\rmd s\eqsp,\\
&=(\delta_0-\alpha)(1-\rme^{-\kappa t})/\kappa + \alpha t + (\eta/\kappa) \int_0^t (1-\rme^{-\kappa (t-s)}) \rmd W^2_s\eqsp.
\end{align*}
Defining $\tilde{W}^1_t \eqdef - (\eta/\kappa) \int_0^t (1-\rme^{-\kappa (t-s)}) \rmd W^2_s+\sigma W_t^1$ and $\tilde{W}_t^2\eqdef \int_0^t \eta \rme^{\kappa(s-t)}\rmd W^2_s$, this yields:
\begin{align*}
X_t &= X_0 +  (\mu-\sigma^2/2)t + (\alpha-\delta_0)(1-\rme^{-\kappa t})/\kappa - \alpha t + \tilde{W}_t^1\eqsp,\\
\delta_t &= \delta_0\rme^{-\kappa t} + \alpha(1-\rme^{-\kappa t}) + \tilde{W}^2_t\eqsp.
\end{align*}
The proof is concluded upon noting that $\tilde{W}^1_t$ and $\tilde{W}^2_t$ are centered Gaussian random variables such that:
\begin{enumerate}[-]
\item $\!\!\mathrm{Var}\left[\tilde{W}^1_t\right] = \sigma^2  t + \eta^2\left\{ t + (1-e^{-2\kappa t})/(2\kappa) - 2(1-e^{-\kappa  t})/\kappa\right\}/\kappa^2 - 2\rho \eta \sigma\left\{ t - (1-e^{-\kappa  t})/\kappa \right\}/\kappa\eqsp,$
\item $\!\!\mathrm{Var}\left[\tilde{W}^2_t\right] = \eta^2(1-\rme^{-2\kappa t})/(2\kappa)\eqsp,$
\item $\!\!\mathrm{Cov}\left[\tilde{W}^1_t,\tilde{W}^2_t\right] = \left(\rho \eta\sigma-\eta^2/\kappa\right)\left(1-e^{-\kappa t} \right)/\kappa + \eta^2 \left(1-e^{-2\kappa t} \right)/(2\kappa^2)\eqsp.$
\end{enumerate}
\end{proof}

\bibliographystyle{alpha}

\newcommand{\etalchar}[1]{$^{#1}$}


\begin{thebibliography}{}

\end{thebibliography}


\begin{thebibliography}{FGDW02}

\bibitem[Alm16]{almansour:2016}
A.~Almansour.
\newblock Convenience yield in commodity price modeling: A regime switching
  approach.
\newblock {\em Energy Economics}, 53:238--247, 2016.

\bibitem[Bar06]{barber:2006}
D.~Barber.
\newblock {Expectation correction for smoothed inference in switching linear
  dynamical systems}.
\newblock {\em Journal of {M}achine {L}earning {R}esearch}, 7:2515--2540, 2006.

\bibitem[BDM10]{briers:doucet:maskell:2010}
M.~Briers, A.~Doucet, and S.~Maskell.
\newblock Smoothing algorithms for state-space models.
\newblock {\em {A}nnals of the {I}nstitute of {S}tatistical {M}athematics}, 62(1):61--89, 2010.

\bibitem[BF63]{bryson:frazier:1963}
A.E. Bryson and M.~Frazier.
\newblock {Smoothing for linear and nonlinear dynamic systems}.
\newblock {\em {Proceedings of the optimum system synthesis conference}}, 1963.

\bibitem[BGM08]{barembruch:garivier:moulines:2008}
S.~Barembruch, A.~Garivier, and E.~Moulines.
\newblock {On optimal sampling for particle filtering in digital
  communication}.
\newblock {\em IEEE 9th Workshop on Signal Processing Advances in Wireless},
  pages 634--638, 2008.

\bibitem[CL00]{chen:liu:2000}
R.~Chen and J.S. Liu.
\newblock Mixture kalman filters.
\newblock {\em Journal of the Royal Statistical Society B}, 62:493--508, 2000.

\bibitem[DGA00]{doucet:godsill:andrieu:2000}
A.~Doucet, S.~Godsill, and C.~Andrieu.
\newblock On sequential monte carlo sampling methods for bayesian filtering.
\newblock {\em Statistics and computing}, 10:197--208, 2000.

\bibitem[DGK01]{doucet:gordon:krishnamurthy:2001}
A.~Doucet, N.~Gordon, and V.~Krishnamurthy.
\newblock Particle filters for state estimation of jump {M}arkov linear
  systems.
\newblock {\em {IEEE T}ransactions on {S}ignal {P}rocessing}, 49(3):613--624,
  2001.

\bibitem[DLR77]{dempster:laird:rubin:1977}
A.~P. Dempster, N.~M. Laird, and D.~B. Rubin.
\newblock Maximum likelihood from incomplete data via the {EM} algorithm.
\newblock {\em Journal of the Royal Statistical Society: Series B}, 39(1):1--38 (with
  discussion), 1977.

\bibitem[DM13]{delmoral:2013}
P.~Del~Moral.
\newblock {\em Mean field simulation for Monte Carlo integration}.
\newblock Chapman \& Hall / CRC Monographs on Statistics \& Applied
  Probability, 2013.

\bibitem[FC03]{fearnhead:clifford:2003}
P.~Fearnhead and P.~Clifford.
\newblock {On-line inference for hidden {M}arkov models via particle filters}.
\newblock {\em {Journal of the Royal Statistical Society: Series B}}, 65:887--899, 2003.

\bibitem[FGDW02]{fong:godsill:doucet:west:2002}
W.~Fong, S.J. Godsill, A.~Doucet, and M.~West.
\newblock {Monte Carlo smoothing with application to audio signal enhancement}.
\newblock {\em IEEE Transactions on {S}ignal {P}rocessing}, 50(2):438--449,
  2002.

\bibitem[FWT10]{fearnhead:wyncoll:tawn:2010}
P.~Fearnhead, D.~Wyncoll, and J.~Tawn.
\newblock {A sequential smoothing algorithm with linear computational cost}.
\newblock {\em {Biometrika}}, 97:447--464, 2010.

\bibitem[GE90]{gibson:schwartz:1990}
R.~Gibson and Schwartz E.S.
\newblock Stochastic convenience yield and the pricing of oil contingent
  claims.
\newblock {\em Journal of finance}, 45(3):959--976, 1990.

\bibitem[HK98]{huerzeler:kunsch:1998}
M.~H{\"u}rzeler and H.~R. K{\"u}nsch.
\newblock {M}onte {C}arlo approximations for general state-space models.
\newblock {\em Journal of {C}omputational and {G}raphical {S}tatistics}, 7:175--193, 1998.

\bibitem[HK04]{hansen:kern:2004}
N.~Hansen and S.~Kern.
\newblock {Evaluating the CMA Evolution Strategy on Multimodal Test Functions}.
\newblock {\em Eighth International Conference on Parallel Problem Solving from
  Nature}, 72:337--354, 2004.

\bibitem[HO01]{hansen:ostermeier:2001}
N.~Hansen and A.~Ostermeier.
\newblock {Completely derandomized self-adaptation in evolution strategies}.
\newblock {\em Evolutionary {C}omputation}, 9(2):159--195, 2001.

\bibitem[Kal60]{kalman:1960}
R.E. Kalman.
\newblock {A new approach to linear filtering and prediction problems}.
\newblock {\em Journal of {B}asic {E}ngineering}, 82:35--45, 1960.

\bibitem[Kim94]{kim:1994}
CJ. Kim.
\newblock Dynamic linear models with markov-switching.
\newblock {\em Journal of Econometrics}, 60(1-2):1--22, 1994.

\bibitem[LBGS13]{lindsten:bunch:godsill:schon:2013}
F.~Lindsten, P.~Bunch, S.J. Godsill, and T.B. Schon.
\newblock Rao-{B}lackwellized particle smoothers for mixed linear/nonlinear
  state-space models.
\newblock {\em Proceedings of the 38th {IEEE I}nternational {C}onference on
  {A}coustics, {S}peech and {S}ignal {P}rocessing ({ICASSP})}, 2013.

\bibitem[LBS{\etalchar{+}}16]{lindsten:bunch:sarkka:schon:godsill:2015}
F.~Lindsten, P.~Bunch, S.~Sarkka, T.B. Schon, and S.J. Godsill.
\newblock Rao-{B}lackwellized particle smoothers for conditionally linear
  {G}aussian models.
\newblock {\em IEEE {J}ournal of {S}elected {T}opics in {S}ignal {P}rocessing},
  10(2):353--365, 2016.

\bibitem[RST65]{rauch:striebel:tung:1965}
A.E. Rauch, C.T. Striebel, and F.~Tung.
\newblock {Maximum likelihood estimates of linear dynamic systems}.
\newblock {\em American Institute of Aeronautics and Astronautics journal},
  3(8):1445--1450, 1965.

\bibitem[Sar13]{sarkka:2013}
S.~Sarkka.
\newblock {\em Bayesian filtering and smoothing}.
\newblock Cambridge {U}niversity {P}ress, 2013.

\bibitem[SBG12]{sarkka:bunch:godsill:2012}
S.~Sarkka, P.~Bunch, and S.J. Godsill.
\newblock A backward-simulation based {R}ao-{B}lackwellized particle smoother
  for conditionally linear {G}aussian models.
\newblock {\em Proceedings of the 16th {IFAC S}ymposium on {S}ystem
  {I}dentification {(SYSID)}}, 2012.

\bibitem[Sch97]{schwartz:1997}
E.S. Schwartz.
\newblock The stochastic behaviour of commodity prices: Implications for
  pricing and hedging.
\newblock {\em The {J}ournal of {F}inance}, 3:923--973, 1997.

\bibitem[SGN05]{schon:gustafsson:nordlund:2005}
T.~Schon, F.~Gustafsson, and P.-J. Nordlund.
\newblock Marginalized particle filters for mixed linear/nonlinear state-space
  models.
\newblock {\em {IEEE T}ransactions on {S}ignal {P}rocessing}, 53(7):2279--2289,
  2005.

\end{thebibliography}
\end{document}